%% file: final-version.tex
\documentclass{lmcs}
\pdfoutput=1

\usepackage{lastpage}
\lmcsdoi{18}{2}{7}
\lmcsheading{}{\pageref{LastPage}}{}{}%
{Oct.~26,~2020}{May~10,~2022}{}


\usepackage[utf8]{inputenc}
\usepackage{amssymb}
\usepackage{amsmath, latexsym}
\usepackage{dsfont}
\usepackage{amsthm,amsfonts}
\usepackage{xspace}
\usepackage{url}
\usepackage{ifthen}
\usepackage{theoremref}

\newcommand\fun[1]{\textit{#1}\xspace}
\newcommand\e{\epsilon}

\input{macro}

\nc{\OMIT}[1]{}

\newcommand\todo[1]{\
  \newline\noindent\bigskip\framebox{\parbox{\columnwidth}{\tt #1}}\bigskip}
\nc{\luc}[1]{\todo{\textbf{Luc:} #1}}
\nc{\arnaud}[1]{\todo{\textbf{Arnaud:} #1}}
\nc{\nicole}[1]{\todo{\textbf{Nicole:} #1}}

\usepackage{etoolbox}
\apptocmd{\sloppy}{\hbadness 10000\relax}{}{} 

\begin{document}

\title[Enumeration of FO queries over low degree databases]{Enumerating Answers to First-Order Queries over \texorpdfstring{\\}{} Databases of Low Degree}

\author[A.~Durand]{Arnaud Durand\rsuper{a}}
\address{IMJ-PRG, Universit\'e de Paris, CNRS}
\author[N.~Schweikardt]{Nicole Schweikardt\rsuper{b}}
\address{Humboldt-Universit\"at zu Berlin}
\author[L.~Segoufin]{Luc Segoufin\rsuper{c}}
\address{INRIA, Laboratoire Cogitamus}


\begin{abstract}
  A class of relational databases has low degree if for all $\delta>0$, all
  but finitely many databases in the class have degree at most
  $n^{\delta}$, where $n$ is the size of the database. Typical examples are
  databases of bounded degree or of degree bounded by $\log n$.

  It is known that over a class of databases having low degree, first-order
  boolean queries can be checked in pseudo-linear time, i.e.\ for all $\epsilon>0$
  in time bounded by $n^{1+\epsilon}$. We generalize this result by
  considering query evaluation.

  We show that counting the number of answers to a query can be done in
  pseudo-linear time and that after a pseudo-linear time preprocessing we can
  test in constant time whether a given tuple is a solution to a query or
  enumerate the answers to a query with constant delay.
\end{abstract}

\maketitle

\section{Introduction}\label{section-introduction}

Query evaluation is a fundamental task in databases and a vast literature is
devoted to the complexity of this problem.  However, for more demanding tasks
such as producing the whole set of answers or computing aggregates on the query
result (such as counting the number of answers), complexity bounds are often
simply extrapolated from those for query evaluation; and until recently, few
specific methods and tools had been developed to tackle these problems.  Given
a database $\A$ and a first-order query $q$, it may be not satisfactory enough
to express complexity results in terms of the sizes of $\A$ and $q$ as it is
often the case. The fact that the solution set $q(\A)$ may be of size
exponential in the query is intuitively not sufficient to make the problem
hard, and alternative complexity measures had to be found for query
answering. In this direction, one way to define tractability is to assume that
tuples of the query result can be generated one by one with some regularity,
for example by ensuring a fixed delay between two consecutive outputs once a
necessary precomputation has been done to construct a suitable index structure.
This approach, that considers query answering as an enumeration problem, has
deserved some attention over the last few years.  In this vein, the best that
one can hope for is constant delay, i.e., the delay depends only on the size of
$q$ (but not on the size of $\A$).  Surprisingly, a number of query evaluation
problems have been shown to admit constant delay algorithms, usually preceded
by a preprocessing phase that is linear or almost linear. This is the case when
queries are evaluated over the class of structures of bounded
degree~\cite{DurandG07, KS11}, over the class of structures of ``bounded
expansion''~\cite{KazanaS13} and, more generally, over the class of nowhere
dense structures~\cite{DBLP:conf/pods/SchweikardtSV18}. Similar results have
been shown for monadic second-order logic over structures of bounded
tree-width~\cite{Courcelle09,Bagan06,WL12} or for fragments of first-order
logic over arbitrary structures~\cite{BDG07,Brault-Baron12}.  However, as shown
in~\cite{BDG07}, the fact that evaluation of boolean queries is easy does not
guarantee the existence of such efficient enumeration algorithms in general:
under some reasonable complexity assumption, there is no constant delay
algorithm with linear preprocessing enumerating the answers of acyclic
conjunctive queries (although it is well-known that the model checking of
boolean acyclic queries can be done in linear time~\cite{Yannakakis81}).

In this paper, we investigate the complexity of the enumeration, counting, and
testing problems for first-order queries over classes of low degree. A class of
relational databases has low degree if for all $\delta>0$, all sufficiently
large databases in the class have degree at most $n^{\delta}$, where $n$ is the
size of the database. Databases of bounded degree or of degree bounded by
$(\log n)^c$, for any fixed constant $c$, are examples of low degree
classes. However, it turns out to be incomparable with the class of nowhere
dense databases mentioned above.

It has been proved in~\cite{Grohe-STACS01} that over a class of databases of
low degree, first-order boolean queries can be checked in pseudo-linear time,
i.e., for all $\epsilon>0$ there is an algorithm running in time bounded by
$O(n^{1+\epsilon})$ checking the given first-order query. In this paper, we
prove that counting the number of answers to a query can be done in
pseudo-linear time, and that enumerating the answers to a query can be done
with constant delay after a pseudo-linear time preprocessing. We also prove
that testing membership of a tuple to a query result can be done in constant
time after a pseudo-linear time preprocessing.  We adopt a uniform approach to
prove all these results by using at the heart of the preprocessing phases a
quantifier elimination method that reduces our different problems to their
analog but for colored graphs and quantifier-free queries. With such a tool, we
can then focus within each specific task on very simple instances.

Over a class of databases of low degree, the
difficulty is to handle queries requiring that in all its answers, some of its
components are far away from each other. When this is not the case, for
instance when in all answers all its components are within short distance from
the first component, then the low degree assumption implies that there are only
few answers in total and those can be computed in pseudo-linear time. In the
difficult case, the number of answers may be exponential in the arity of the
query and the naive evaluation algorithm may spend too much time processing
tuples with components close to each other. To avoid this situation, we
introduce suitable functions that can be precomputed in pseudo-linear time, and
that allow us to jump in constant time from a tuple with components close to
each other to a correct answer.

\textbf{Related work.} \
Enumerating the answers to a boolean query $q$ over a database \A is more general than
testing whether $q$ holds on \A, a problem also known as the model checking
problem. An enumeration algorithm with constant delay after a preprocessing
phase taking 
pseudo-linear time, or even polynomial time,
induces a model checking algorithm that is \emph{fixed-parameter
  tractable} (FPT), i.e, works in time $f(q){\cdot}\size{\A}^c$ for some constant $c$
and some function $f$ depending only on the class of databases.  There is a
vast literature studying the model checking problem for first-order logic
aiming at finding FPT algorithms for larger and larger classes of databases.
Starting from classes of databases of bounded degree, or bounded
treewidth, FPT algorithms were derived for classes of databases having bounded
expansion~\cite{DvorakKralThomas10} (see also~\cite{KazanaS13}). Actually, recently an FPT algorithm has been
obtained for classes of databases known as
``nowhere dense'', generalizing all the previously known results~\cite{GKS13}.
This last result is in a sense ``optimal'' as it is known that if a class of
databases is closed under substructures and has no FPT model checking algorithm
then it is somewhere dense~\cite{KD09}, modulo some reasonable complexity
hypothesis.

Classes of databases of low degree do not belong to this setting. It is easy to
see that they are neither nowhere dense nor closed under substructures (see
Section~\ref{subsection:LowDegree}).  Our algorithms build on the known model
checking algorithm for low degree databases~\cite{Grohe-STACS01}. They
generalize the known enumeration algorithms for databases of bounded
degree~\cite{DurandG07, KS11}.

This paper is the journal version of~\cite{DBLP:conf/pods/DurandSS14}. There
is an important difference with the conference version. In the conference version
we needed the extra hypothesis that even though we would use only a memory of
pseudo-linear size, a total amount of memory of quadratic size was necessary
for our algorithms to work. This extra memory is no longer necessary here
thanks to the data structure constructed in
Theorem~\ref{thm-storing-complete}. This makes the technical lemma slightly
more complicated to state, but does not affect the general results.

\medskip
\textbf{Organization.} \
We fix the basic notation and formulate our main results in
Section~\ref{section-prelim}.
In Section~\ref{section-main} we present the algorithms for counting,
testing, and enumerating answers to first-order queries over classes of
structures of low degree. These algorithms rely on a particular
preprocessing which transforms a first-order query on a database
into a quantifier-free query on a colored graph. The
result is stated in Section~\ref{subsection:qeli}, while its proof is presented in Section~\ref{section-main-proofs}.
We conclude in Section~\ref{section-conclusion}.

\section{Preliminaries and Main Results}\label{section-prelim}

We write $\NN$ to denote the set of non-negative integers, and we let
$\NNpos\deff\NN\setminus\set{0}$.
$\QQ$ denotes the set of  rationals, and $\QQpos$ is the set of positive  rationals.

\subsection{Databases and queries}
A database is a finite relational structure.  A \emph{relational signature} $\sigma$ is
a finite set of relation symbols $R$, each of them associated with a fixed
\emph{arity} $\ar(R)\in\NNpos$. A \emph{relational structure} \A over $\sigma$, or a
$\sigma$-structure (we omit to mention $\sigma$ when it is clear from the
context) consists of a non-empty finite set $\dom(\A)$
called the \emph{domain} of \A, and an $\ar(R)$-ary relation $R^\A\subseteq
\dom(\A)^{\ar(R)}$ for each relation symbol $R\in\sigma$.

The degree of a structure \A, denoted $\degree(\A)$, is the degree of the
Gaifman graph associated with \A (i.e., the undirected graph with vertex set
$\dom(\A)$ where there is an edge between two nodes if they both occur in a
tuple that belongs to a relation of \A).
With this definition, in a structure with
$n$ domain elements and of degree $d$, each $r$-ary relation may have
at most $n\cdot (d{+}1)^{r-1}$ tuples.

In the sequel we only consider structures of degree $\geq 2$. As structures
of degree~1 are quite trivial, this is without loss of generality.

We define the \emph{size} $\size{\A}$ of $\A$ as $\size{\A}=
|\sigma|+|\dom(\A)|+\sum_{R\in\sigma} |R^{\A}|{\cdot} \ar(R)$. It corresponds
to the size of a reasonable
encoding of \A.
We assume that the input structure \A is presented in a
  way such that given a relation symbol $R$ in the signature we can directly access
  the list of tuples in $R^\A$ (i.e.\ without reading the remaining
  tuples).
The cardinality of \A, i.e.\ the cardinality
of its domain, is denoted by $|\A|$.

By \emph{query} we mean a formula of $\FO(\sigma)$, the set of all first-order
formulas of signature $\sigma$, for some relational signature $\sigma$ (again
we omit $\sigma$ when it is clear from the context). For $\varphi\in
\FO$, we write $\varphi(\bar x)$ to denote a query whose free variables
are $\bar x$, and the number of free variables is called the \emph{arity of the
  query}.  A \emph{sentence} is a query of arity 0.  Given a structure \A and a
query $\varphi$, an \emph{answer} to $\varphi$ in \A is a tuple $\bar a$ of
elements of $\dom(\A)$ such that $\A \models \varphi(\bar a)$.
In the special case where $\varphi$ is a sentence, it is either true of
false in \A, and the former is denoted $\A \models \varphi$ and the
latter is denoted $\A\not\models\varphi$.
We write $\varphi(\A)$
for the set of answers to $\varphi$ in \A, i.e. $\varphi(\A)=\setc{\bar a}{\A
  \models \varphi(\bar a)}$. As usual, $|\varphi|$ denotes the size of
$\varphi$.

Let \Class be a class of structures. The model checking problem of \FO over \Class is
the computational problem of given a {\bf sentence} $\varphi\in\FO$ and a database
$\A\in\Class$ to test whether $\A\models \varphi$. 

Given a $k$-ary query $\varphi$, we care about ``enumerating'' $\varphi(\A)$ efficiently. Let
\Class be a class of structures. The \emph{enumeration
  problem of $\varphi$ over \Class} is, given a database $\A\in\Class$, to output the
elements of $\varphi(\A)$ one by one with no repetition. The time needed to output
the first solution is called the \emph{preprocessing time}. The maximal time
between any two consecutive outputs of elements of $\varphi(\A)$ is called \emph{the
  delay}. We are interested here in enumeration algorithms with pseudo-linear
preprocessing time and constant delay. We now make these notions formal.

\subsection{Model of computation and enumeration}\label{subsection:RAM}
We use Random Access Machines (RAMs) with addition and uniform cost measure as
a model of computation. For further details on this model and its use in logic
see~\cite{FlumGrohe-ParameterizedComplexity,GrandjeanOlive04}.

In the sequel we assume that the input relational structure comes with a linear
order on the domain. If not, we use the one induced by the encoding of the
structure as an input to the RAM\@.
Whenever we iterate through all nodes of the domain, the iteration is with
respect to the initial linear order. The linear order induces a lexicographical
order on tuples of elements of the domain.

Our algorithms over RAMs will take as input a query $\varphi$ of size $k$ and a
structure \A of size $n$.  We adopt the data complexity point of view and say that a problem can be solved in \emph{linear
  time} (respectively, \emph{constant time}) if it can be solved by an
algorithm outputting the solution within $f(k){\cdot} n$ steps (respectively, $f(k)$ steps), for some function $f$. We
also say a problem can be solved in \emph{pseudo-linear time} if, for all
$\epsilon\in \QQpos$, there is an algorithm solving it within $f(k,\epsilon){\cdot}
n^{1+\epsilon}$ steps, for some function $f$.

We will often compute partial $k$-ary functions $f$ associating a value to a tuple
of nodes of the input graph. Such functions can be easily implemented in the
RAM model using $k$-dimensional cubes allowing to retrieve the value of $f$ in
constant time.  This requires a memory usage of $O(n^k)$ and an initialization
process of $O(n^k)$.  However our functions will have a domain of size
pseudo-linear and can be computed in pseudo-linear time.  The following theorem
states that we can use the RAM model to build a data structure that stores our
functions in a more efficient way.
{The data structure is a trie  of depth $\frac{1}{\e}$ and of
  degree $n^\e$ where each pair (key,value) is a tuple $\bar a$ in the domain of $f$ and its
  image $b=f(\bar a)$.  The details can be found
  in~\cite{DBLP:conf/pods/SchweikardtSV20}.}

\begin{thm}[Storing Theorem]\label{thm-storing-complete}
  For every fixed $n,k\in\NN$ and $\epsilon>0$, there is a data
  structure that stores the value of a $k$-ary function $f$ of domain
  $\dom(f)\subseteq [n]^k$ with:
	\begin{itemize}
		\item computation time and storage space $O(|\dom(f)|\cdot n^{\e})$,
		\item lookup time only depending on $k$ and $\e$,
	\end{itemize}
	Here, lookup means that given a tuple
$\bar a\in [n]^k$,
the algorithm either answers $b$ if $\bar a\in\dom(f)$ and $f(\bar a)=b$, or \void otherwise.
\end{thm}

An important consequence of Theorem~\ref{thm-storing-complete} is that, modulo
a preprocessing in time pseudo-linear in the size of the database, we can test
in constant time whether an input tuple is a fact of the database:

\begin{cor}\label{cor-test}
  Let \A be a database over the schema $\sigma$. Let $\e >0$. One can
  compute in time $O(d^r n^{1+\e})$ a
  data structure such that on input of a tuple $\bar a$ and a relation symbol $R
  \in \sigma$ one can test whether  $\A \models R(\bar a)$ in time $O(1)$, where
  $n=|\dom(\A)|$, $d=\degree(\A)$, and $r$ is a number that only
  depends on $\sigma$.
\end{cor}
\begin{proof}
 Immediate from Theorem~\ref{thm-storing-complete}, as the number of
 tuples in an $r$-ary relation of
  a $\sigma$-structure whose Gaifman graph has degree $d$ is at most $(d{+}1)^{r-1}n$.
\end{proof}

Note that a simple linear time preprocessing would provide a data structure
allowing for a test as in Corollary~\ref{cor-test} in time $O(d)$
{(to see how, just think of the special case where $\A$ only
  contains one binary relation $R$. In this case, the preprocessing can
  build an adjacency list representation where upon input of an $a$ we
  can access in time $O(1)$ the first element of a list of all those
  $b$ satisfying $R(a,b)$. This list has length at most $d$. Upon input of a tuple $(a,b)$ we
  access $a$'s adjacency list and check if it contains $b$).}
{With
  the help of the
Storing Theorem we get a test in constant time,}
i.e.\ depending only on $\e$ and $\sigma$ and not on $d$.

We say that the \emph{enumeration problem} of \FO over a class \Class of structures
can be solved with \emph{constant delay after a pseudo-linear preprocessing}, if
it can be solved by a RAM algorithm which, on input $\e >0$, $q\in\FO$ and $\A\in\Class$, can be
decomposed into two phases:

\begin{itemize} \itemsep1pt \parskip0pt \parsep0pt
	\item a preprocessing phase that is performed in time $f(\e,|q|)\cdot
          \size{\A}^{1+\epsilon}$ for some function $f$, and
	\item an enumeration phase that outputs $q(\A)$ with no repetition and
          a delay depending only on $q$, $\e$, and \Class between any two consecutive
          outputs.
          The enumeration phase has full access to the output of the preprocessing
          phase and can use extra memory whose size depends only on $q$, $\e$ and \Class.
\end{itemize}

\noindent
Notice that if we can enumerate $q$ with constant delay after a pseudo-linear
preprocessing, then all answers can be output in time
$g(|q|,\epsilon)\cdot(\size{\A}^{1+\epsilon}+|q(\A)|)$, for some function $g$,
and the first solution is computed in pseudo-linear time. In the particular case
of boolean queries, the associated model checking problem must be solvable in
pseudo-linear time.

\begin{exa}\label{example-def}
  To illustrate these notions, consider the binary query $q(x,y)$ over colored graphs
  computing the pairs of nodes $(x,y)$ such that $x$ is \emph{blue}, $y$ is
  \emph{red}, and there is no edge from $x$ to $y$. It can be expressed in
  \FO by
  \begin{equation*}
  B(x) \land R(y) \land \lnot E(x,y).
  \end{equation*}
  A naive algorithm for evaluating $q$ would iterate through all blue nodes,
  then iterate through all red nodes, check if they are linked by an edge and,
  if not, output the resulting pair, otherwise try the next pair.

  With our RAM model, after a linear preprocessing, we can easily iterate
  through all blue nodes with a constant delay
  between any two of them and similarly for red nodes. By
  Corollary~\ref{cor-test}, we can test in constant
  time whether there is an edge between any two nodes. The problem with this
  algorithm is that many pairs of appropriate color may be false
  hits. Hence the delay between two consecutive outputs may be arbitrarily large.

  If the degree is assumed to be bounded by a fixed constant, then the above
  algorithm enumerates
  all answers with constant delay, since the number of false hits for each blue
  node is bounded by the degree. We will see that for structures of low degree
  we can modify the algorithm in order to achieve the same result.
\end{exa}

\subsection{Classes of structures of low degree}\label{subsection:LowDegree}

Intuitively a class \Class of structures has \emph{low
degree} if for all $\delta>0$, all but finitely many structures \A of \Class have degree
at most $|\A|^\delta$ (see~\cite{Grohe-STACS01}).
More formally, $\Class$  has low degree if for
  every $\delta \in \QQpos$ there is an $n_\delta\in\NNpos$ such that all structures
  $\A\in\Class$ of cardinality $|\A|\geq n_\delta$ have $\degree(\A)\leq |\A|^\delta$.
If there is a computable function associating $n_\delta$ from $\delta$ then we
furthermore say that the class is effective.

For example, for every fixed number $c>0$, the class of all
structures of degree at most $(\log n)^c$ is of low degree and effective.
Clearly, an arbitrary class $\Class$ of structures can be
transformed into a class $\Class'$ of low degree by padding each
$\A\in \Class$ with a suitable number of isolated elements (i.e.,
elements of degree 0). Therefore classes of low degree
are usually \emph{not} closed under taking substructures. In particular if we
apply the padding trick to the class of cliques, we obtain a class of low
degree that is not in any of the classes with known low evaluation complexity
such as the ``nowhere dense'' case mentioned in the introduction.

Notice that $\degree(\A)\leq |\A|^\delta$ implies that $\size{\A}\leq
c{\cdot}|\A|^{1+\delta\cdot r}$, where $r$
is the maximal arity of the signature and $c$ is a number only
depending on $\sigma$.

It is known that on classes of graphs of low degree, model
checking of first-order sentences can be done in pseudo-linear time. We will
actually need the following stronger result:

\begin{thm}[Grohe~\cite{Grohe-STACS01}]\label{thm:grohe-low-degree}
  There is a computable function $h$ such that on input of a structure $\A\in\Class$ of
  degree $d$ and a sentence
  $q\in\FO$, one can test in time   $h(|q|){\cdot}
  |\A| {\cdot} d^{h(|q|)}$ whether $\A\models q$.

In particular, if $\Class$ is a class of structures of low degree,
then there is a function $g$ such that, given a structure
$\A\in\Class$, a sentence $q\in\FO$, and $\epsilon>0$, one can check
if $\A\models q$ in time
$g(|q|,\epsilon){\cdot}|\A|^{1+\epsilon}$. If $\Class$ is effective
then $g$ is computable.
\end{thm}

\subsection{Our results}

We are now ready to state our main results, which essentially lift
Theorem~\ref{thm:grohe-low-degree} to non-boolean queries and to counting,
testing, and enumerating their answers.

Our first result is that we can count the number of answers to a query in pseudo-linear time.

\begin{thm}\label{thm:counting}
  Let \Class be a class of structures of low degree.  There is a function $g$
  such that, given a structure $\A\in\Class$, a query $q\in\FO$, and $\e>0$, one
  can compute $|q(\A)|$ in
  time $g(|q|,\epsilon){\cdot} |\A|^{1+\epsilon}$. If \Class is effective then g is computable.
\end{thm}

Our second result is that we can test whether a given tuple is part of the
answers in constant time after a pseudo-linear time preprocessing.

\begin{thm}\label{thm:testing}
  Let \Class be a class of structures of low degree.  There is a function $g$
  such that, given a structure $\A\in\Class$, a query $q\in\FO$, and $\e>0$, one
  can compute in time $g(|q|,\epsilon){\cdot} |\A|^{1+\epsilon}$ a data structure such that, on input of any $\bar a$, one can then
  test in time $g(|q|,\epsilon)$ whether $\bar a\in q(\A)$. If \Class is effective then g is computable.
\end{thm}

Finally, we show that we can enumerate the answers to a query with constant
delay after a pseudo-linear time preprocessing.

\begin{thm}\label{thm:enum}
  Let \Class be a class of structures of low degree. There is a function $g$
  such that, given a structure $\A\in \Class$, a query $q\in \FO$ and
  $\e>0$, the enumeration problem of $q$ over \A can be solved with delay $g(|q|,\epsilon)$
  after a preprocessing running in time $g(|q|,\epsilon){\cdot}
  |\A|^{1+\epsilon}$.
 If \Class is effective then g is computable.
\end{thm}

\subsection{Further notation}

We close this section by fixing technical notations that will be used
throughout this paper.

For a structure $\A$ we write $\dist^\A(a,b)$ for the distance between
two nodes $a$ and $b$ of the Gaifman graph of $\A$.
For an element $a\in\dom(\A)$ and a number $r\in\NN$, the
\emph{$r$-ball} around $a$ is the set $N_r^\A(a)$ of all nodes
$b\in\dom(\A)$ with $\dist^\A(a,b)\leq r$.
The \emph{$r$-neighborhood} around $a$ is the induced substructure
$\N_r^\A(a)$ of $\A$ on $N_r^\A(a)$.
Note that if $\A$ is of degree $\leq d$ for $d\geq 2$, then
$|N_r^\A(a)|\leq \sum_{i=0}^r d^i < d^{r+1}$.

\section{Evaluation algorithms}\label{section-main}

In this section, we present our algorithms for counting, testing, and enumerating
the solutions to a query (see Sections~\ref{subsection:Counting},~\ref{subsection:Testing}, and~\ref{subsection:Enumeration}). They all build on the
same preprocessing algorithm which
runs in pseudo-linear time and which essentially reduces the input to a
quantifier-free query over a  suitable signature (see Section~\ref{subsection:qeli}).
However, before presenting these algorithms, we start with very simple cases.

  \subsection{Computing the neighborhoods}

Unsurprisingly, all our algorithms will start by computing the neighborhood
$\N_r^\A(a)$ for all the elements $a$ of the input structure \A for a suitable
constant $r$ depending only on $q$. We actually do not need to compute
$\N^\A_r(a)$ but only
$\N^{\A\downarrow q}_r(a)$ where $\Aq$ is the restriction
of $\A$ to the relational symbols occurring in $q$.

The next lemma states that all these
neighborhoods can be computed in reasonable time.

\begin{lem}\label{lemma-compute-neigh}
  There is an algorithm which, at input of a structure \A, a query $q$ and a
  number $r$ computes $\N^{\A\downarrow q}_r(a)$ for all elements $a$ of the
  domain of \A in time $O(|q| \cdot n \cdot d^{h(r,|q|)})$, where
  $n=|\dom(\A)|$, $d=\degree(\A)$ and $h$ is a computable function.
\end{lem}
\begin{proof}
  We first compute the Gaifman graph associated to $\Aq$. We consider
  each relation symbol $R$ occurring in $q$. We scan $R^\A$ and for each of its tuples
  we add the corresponding clique to the Gaifman graph. As each element $a$ of
  $\dom(\A)$ can appear at the first component of at most $(d+1)^{\ar(R)-1}$ tuples of
  $R^\A$, the total time
  is $|q| \cdot n \cdot (d+1)^{k-1}$ where
  $k$ is
  the maximal arity of the relation symbols occurring in $q$.
    This yields a graph with $n$ nodes and
  degree $d$ corresponding to the Gaifman graph of $\Aq$. From this
  graph we easily derive, in time $O(n \cdot d)$,  a structure associating to each
  node $a \in \dom(\A)$ the set of its immediate neighbors, i.e $N_1^{\A\downarrow q}(a)$.

  With $r$ steps of transitive closure computation we get a structure
  associating to each node $a\in\dom(\A)$ the set $N^{\A\downarrow q}_r(a)$ of
  nodes at distance at most $r$ from $a$. This can be done
  in time $O(n\cdot d^r)$.

  With an extra scan over the database we can derive $\N_r^{\A\downarrow q}(a)$
  from $N^{\A\downarrow q}_r(a)$. Altogether we get the desired time bounds.
\end{proof}

  \subsection{Connected conjunctive queries}\label{subsection:warm-up}

  As a warm-up for working with classes of structures of low degree, we first
  consider the simple case of queries which we call
  \emph{connected conjunctive queries}, and which are defined as follows.

  A \emph{conjunction} is a query $\gamma$ which is a conjunction of relational
  atoms and potentially negated \emph{unary} atoms.  Note that the query of
  Example~\ref{example-def} is not a conjunction as it has a binary negated
  atom.  With each conjunction $\gamma$ we associate a \emph{query graph}
  $H_\gamma$. This is the undirected graph whose vertices are the variables
  $x_1,\ldots,x_k$ of $\gamma$, and where there is an edge between two vertices
  $x_i$ and $x_j$ iff $\gamma$ contains a relational atom in which both $x_i$
  and $x_j$ occur.
  We call the conjunction $\gamma$ \emph{connected} if its query graph $H_\gamma$ is
  connected.

  A \emph{connected conjunctive query} is a query
  $q(\ov{x})$ of the form $\exists \ov{y}\, \gamma(\ov{x},\ov{y})$, where
  $\gamma$ is a \emph{connected conjunction} in which all variables of
  $\ov{x},\ov{y}$ occur (here, $|\ov{y}|=0$ is allowed).

  The next simple lemma implies that over a class of structures of low degree,
  connected conjunctive queries can be evaluated in pseudo-linear time. It will
  be used in several places throughout this paper: in the proof of
  Proposition~\ref{prop:connected-cqs_low-degree}, and in the proofs for our
  counting and enumeration results in
  Sections~\ref{subsection:Counting} and~\ref{subsection:Enumeration}.

  \begin{lem}\label{lemma:connected-cqs}
   There is an algorithm which, at input of a structure $\A$ and a
   connected conjunctive query $q(\ov{x})$ computes $q(\A)$ in time
   $O(|q|{\cdot}n{\cdot}d^{h(|q|)})$, where
   $n=|\dom(\A)|$,
   $d=\degree(\A)$, and
   $h$ is a computable function.
  \end{lem}
  \begin{proof}
    Let $q(\ov{x})$ be of the form $\exists \ov{y}\,
    \gamma(\ov{x},\ov{y})$, for a connected conjunction $\gamma$.
    Let $k=|\ov{x}|$ be the number of free variables of $q$, let
    $\ell=|\ov{y}|$, and let
    $r=k+\ell$.

    In view of Lemma~\ref{lemma-compute-neigh} we can assume that all the
    neighborhoods $\N^{\A\downarrow q}(a)$ have been computed. In order to
    simplify the notations, in the rest of
    the proof we will write $\N_r^{\A}(a)$ instead of
    $\N_r^{\A\downarrow q}(a)$.

    Note that since $\gamma$ is connected, for every tuple $\ov{c}\in\gamma(\A)$ the
    following is true, where $a$ is the first component of $\ov{c}$.
    All components $c'$ of $\ov{c}$ belong to the $r$-neighborhood
    $\N^\A_r(a)$ of $a$ in $\dom(\A)$.
    Thus, $q(\A)$ is the disjoint union of the sets
    \[
       S_a \ \deff \
       \big\{\, \ov{b} \in q(\N^{\A}_r(a)) \, : \, \text{the first component of
           $\ov{b}$ is $a$}\,\big\},
    \]
    for all $a\in \dom(\A)$.
    For each $a\in \dom(\A)$, the set $S_a$ can be computed as follows:
    \begin{me}
     \item
       Initialize $S_a\deff \emptyset$.

     \item
       Use a brute-force algorithm that enumerates all $k$-tuples
       $\ov{b}$ of elements in $\N^{\A}_r(a)$ whose first component is $a$.

       For each such tuple $\ov{b}$, use a brute-force algorithm
       that checks whether $\N^{\A}_r(a) \models q(\ov{b})$.
       If so, insert $\ov{b}$ into $S_a$

       Note that the number of considered tuples
       $\ov{b}$ is $\leq d^{(r+1)(k-1)}$.
       And checking whether $\N^{\A}_r(a) \models q(\ov{b})$ can be
       done in time $O(|\gamma|{\cdot}d^{(r+1)\ell})$: for this, enumerate all
       $\ell$-tuples $\ov{c}$ of elements in $\N_r^{\A}(a)$ and take
       time $O(|\gamma|)$ to check whether $\gamma(\ov{x},\ov{y})$ is satisfied by the
       tuple $(\ov{b},\ov{c})$.

       Thus, we are done after $O(|\gamma|{\cdot} d^{r^2})$
       steps.
    \end{me}
    In summary, we can compute $q(\A) = \bigcup_{a\in A} S_a$ in time
    $O(n{\cdot}|q|{\cdot} d^{h(|q|)})$, for a computable function~$h$.
  \end{proof}

  As an immediate consequence we can compute in pseudo-linear time the answers
  to a connected conjunctive query over a class of structures of low degree.

  \begin{prop}\label{prop:connected-cqs_low-degree}
   Let $\Class$ be a class of structures of low degree.
   Given a structure $\A\in\Class$, a connected conjunctive query
   $q$, and $\e>0$, one can compute $q(\A)$ in time $g(|q|,\e)\cdot |\A|^{1+\e}$
   for some function $g$ which is computable when $\Class$ is effective.
  \end{prop}

  \begin{proof}
   We use the algorithm provided in Lemma~\ref{lemma:connected-cqs}.
   To see that the running time is as claimed, we use the assumption
   that $\Class$ is of low degree: for every $\delta>0$ there is
   an $m_\delta\in\NNpos$ such that every structure $\A\in \Class$ of
   cardinality $|\A|\geq m_\delta$ has $\degree(\A)\leq |\A|^\delta$.

   For a given $\epsilon>0$ we let $\delta\deff \frac{\epsilon}{h(|q|)}$
   and define $n_\epsilon\deff m_\delta$.
   Then, every $\A\in\Class$ with $|\A|\geq n_\epsilon$ has
   $\degree(\A)\leq |\A|^{\epsilon/h(|q|)}$.
   Thus, on input of $\A$ and $q$, the algorithm from
   Lemma~\ref{lemma:connected-cqs} has running time $O(|q|{\cdot}
   |\A|^{1+\epsilon})$ if $|\A| \geq n_\epsilon$ and takes time bounded by
   $O(|q|{\cdot} n_\epsilon^{1+h(q)})$ otherwise. This gives the bounds claimed
   by the proposition with a computable function $g$ as soon as we can compute $n_\e$.
  \end{proof}

  The method of the proof of
  Proposition~\ref{prop:connected-cqs_low-degree} above will be used
 several times in the paper.

\subsection{Quantifier elimination and normal form}\label{subsection:qeli}

In this section, we make precise the quantifier elimination approach
that is at the heart of the
preprocessing phase of the query evaluation algorithms of our paper.

A signature is \emph{binary} if all its relation symbols have arity at most 2.
A \emph{colored graph} is a finite relational structure over a binary
signature.

\begin{prop}~\label{prop-quantifier-elim-detailed}
 There is an algorithm which, at input of a structure $\A$, a
 first-order query $\varphi(\ov{x})$, and $\e>0$, produces
 a binary signature $\tau$ (containing, among other symbols, a binary
 relation symbol $E$), a colored graph $\G$ of signature $\tau$,
 an $\FO(\tau)$-formula $\psi(\ov{x})$, a mapping $f$, and a data structure such that
 the following is true for $k=|\ov{x}|$, $n=|\dom(\A)|$, $d=\degree(\A)$ and $h$ some computable function:
 \begin{enumerate}
 \item\label{item:psi-form:prop-quantifier-elim-detailed}
   $\psi$  is quantifier-free.
   Furthermore, $\psi$ is of the form
   \,$(\psi_{1} \und \psi_{2})$, where $\psi_{1}$ states that no
   distinct free variables of $\psi$ are connected by an $E$-edge,
   and $\psi_{2}$ is a positive boolean combination of unary atoms.

 \item\label{item:psi-comput:prop-quantifier-elim-detailed}
   $\tau$ and $\psi$ are computed in time and space
   $h(|\varphi|){\cdot} n{\cdot} d^{h(|\varphi|)}$.\\
   Moreover, $|\tau|\leq h(|\varphi|)$ and
   $|\psi|\leq h(|\varphi|)$.

 \item\label{item:G:prop-quantifier-elim-detailed}
   $\G$ is computed in time and space
   $h(|\varphi|){\cdot} n{\cdot} d^{h(|\varphi|)}$.
   \\
   Moreover, 
     $\degree(\G)\leq d^{h(|\varphi|)}$.

 \item\label{item:f:prop-quantifier-elim-detailed}
   $f$ is an injective mapping from $\dom(\A)^k$ to $\dom(\G)^k$
   such that $f$ is a bijection between $\varphi(\A)$ and
   $\psi(\G)$.

   The data structure representing $f$ can be computed in time and
   space $h(|\varphi|){\cdot}
   n^{1+\e}{\cdot} d^{h(|\varphi|)}$ and can then be used as follows: on input of any tuple
   $\bar a \in \dom(\A)^k$, the tuple
   $f(\ov{a})$ can be computed in time $O(k^2)$; and on input of any tuple $\ov{v}\in
   \psi(\G)$, the tuple $f^{-1}(\ov{v})$ can be computed in time $O(k^2)$.

 \end{enumerate}
\end{prop}

\noindent
The proof of Proposition~\ref{prop-quantifier-elim-detailed} is long and
technical and of a somewhat different nature than the results we now
describe. It is postponed to Section~\ref{section-main-proofs}. However, this
proposition is central in the proofs of the results below.

\subsection{Counting}\label{subsection:Counting}

Here we consider the problem of counting the number of solutions to a query on low degree structures.

A \emph{generalized conjunction} is a conjunction of relational atoms
and negated relational atoms (hence, also atoms of arity bigger than
one may be negated, and the query of Example~\ref{example-def} is a generalized
conjunction).

\begin{exa}\label{example-counting}
  Before moving to the formal proof of Theorem~\ref{thm:counting}, consider
  again the query $q$ from Example~\ref{example-def}. Recall that it computes
  the pairs of blue-red nodes that are not connected by an edge. To count its number of
  solutions over a class of structures of low degree we can proceed as
  follows.  We first consider the query $q'(x,y)$ returning the set of
  blue-red nodes that \emph{are} connected. In other words, $q'$ is
  \begin{equation*}
 B(x) \land R(y) \land E(x,y).
  \end{equation*}
  Notice that this query is a connected conjunction. Hence, by
  Proposition~\ref{prop:connected-cqs_low-degree} its answers can be computed
  in pseudo-linear time and therefore we can also count its number of solutions
  in pseudo-linear time. It is also easy to compute in pseudo-linear time the
  number of pairs of blue-red nodes. The number of answers to $q$ is then the
  difference between these two numbers.
\end{exa}

The proof sketch for Theorem~\ref{thm:counting} goes as follows. Using
Proposition~\ref{prop-quantifier-elim-detailed} we can assume modulo a
pseudo-linear preprocessing that our formula is quantifier-free and over a
binary signature. Each connected component is then treated separately and we
return the product of all the results. For each connected component we
eliminate the negated symbols one by one using the trick illustrated in
Example~\ref{example-counting}. The resulting formula is then a connected
conjunction that is treated in pseudo-linear time using
Proposition~\ref{prop:connected-cqs_low-degree}.

\begin{lem}\label{lemma:counting-graphs}
 There is an algorithm which, at input of a colored graph $\G$ and a
 generalized conjunction $\gamma(\ov{x})$, computes
 $|\gamma(\G)|$ in time $O(2^m{\cdot}|\gamma|{\cdot}n{\cdot}d^{h(|\gamma|)})$,
 where
 $h$ is a computable function,
 $m$ is the number of negated binary atoms in $\gamma$,
 $n=|\dom(\G)|$, and
 $d=\degree(\G)$.
\end{lem}
\begin{proof}
By induction on the number $m$ of \emph{negated} binary atoms
  in $\gamma$.
  The base case for $m{=}0$ is obtained as follows.
   We start by using $O(|\gamma|)$ steps to compute the query
    graph $H_\gamma$ and to compute the connected components of
    $H_\gamma$.

    In case that $H_\gamma$ is connected, we can use Lemma~\ref{lemma:connected-cqs}
    to compute the entire set $\gamma(\G)$ in time
    $O(|\gamma|{\cdot}n{\cdot}d^{h(|\gamma|)})$, for a computable
    function $h$. Thus, counting $|\gamma(\G)|$ can be done within the
    same time bound.

    In case that $\gamma$ is not connected,
    let $H_1,\ldots,H_\ell$ be the connected components.
    For each $i\in\set{1,\ldots,\ell}$ let $\ov{x}_i$ be the tuple
    obtained from $\ov{x}$ by removing all variables that do not belong
    to $H_i$.
    Furthermore, let $\gamma_i(\ov{x}_i)$ be the conjunction
    of all atoms or negated unary atoms of $\gamma$ that contain variables in
    $H_i$.
    Note that $\gamma(\ov{x})$ is equivalent to
    $\Und_{i=1}^\ell \gamma_i(\ov{x}_i)$,
    and
    \[
    |\gamma(\G)| \ \ = \ \ \prod_{i=1}^\ell |\gamma_i(\G)|.
    \]

    Since each $\gamma_i$ is connected, we can compute $|\gamma_i(\G)|$
    in time $O(|\gamma_i|{\cdot}n{\cdot}d^{h(|\gamma_i|)})$ by using the
    algorithm of Lemma~\ref{lemma:connected-cqs}.
    We do this for each $i\in\set{1,\ldots,\ell}$ and output the
    product of the values. In summary, we are done in time
    $O(|\gamma|{\cdot}n{\cdot}d^{h(|\gamma|)})$ for the base case $m=0$.

  For the induction step, let $\gamma$ be a formula with $m{+}1$
  negated binary atoms. Let $\nicht R(x,y)$ be a negated binary atom
  of $\gamma$, and let $\gamma_1$ be such that
  \begin{eqnarray*}
     \gamma & \ = \ & \gamma_1 \ \und \ \nicht R(x,y),
     \qquad\text{and let}
     \\
     \gamma_2 & \ \deff \ & \gamma_1 \ \und \ R(x,y).
  \end{eqnarray*}
  Clearly, $|\gamma(\G)| = |\gamma_1(\G)| - |\gamma_2(\G)|$.
  Since each of the formulas $\gamma_1$ and $\gamma_2$ has only
  $m$ negated binary atoms, we can use the induction hypothesis to
  compute $|\gamma_1(\G)|$ and $|\gamma_2(\G)|$ each in time
  $O(2^m{\cdot}|\gamma|{\cdot}n{\cdot}d^{h(|\gamma|)})$.
  The total time used for computing $|\gamma(\G)|$ is thus
  $O(2^{m+1}{\cdot}|\gamma|{\cdot}n{\cdot}d^{h(|\gamma|)})$.
\end{proof}

By using Proposition~\ref{prop-quantifier-elim-detailed}, we can lift
this to arbitrary structures and first-order queries:

\begin{prop}\label{prop:counting}
 There is an algorithm which at input of a structure $\A$ and a
 first-order query $\varphi(\ov{x})$ computes $|\varphi(\A)|$ in time
 $h(|\varphi|){\cdot}n{\cdot}d^{h(|\varphi|)}$, for a computable
 function $h$, where $n=|\dom(\A)|$ and $d=\degree(\A)$.
\end{prop}
\begin{proof}
 We first use the algorithm of
 Proposition~\ref{prop-quantifier-elim-detailed}
 to compute the according graph $\G$ and the quantifier-free formula
 $\psi(\ov{x})$.
 This takes time $h(|\varphi|){\cdot}n{\cdot}d^{h(|\varphi|)}$
 for a computable function $h$.
 And we also know that $|\psi|\leq h(|\varphi|)$.
%
 By Proposition~\ref{prop-quantifier-elim-detailed}
 we know that $|\varphi(\A)|=|\psi(\G)|$.

 Next, we transform $\psi(\ov{x})$ into disjunctive normal form
 \[
   \Oder_{i\in I}\gamma_i(\ov{x}),
 \]
 such that the conjunctive clauses
 $\gamma_i$ exclude each other
 (i.e., for each $\ov{v}\in\psi(\G)$ there is exactly one $i\in I$
 such that $\ov{v}\in \gamma_i(\G)$).
 Clearly, this can be done in time $O(2^{|\psi|})$. Each $\gamma_i$
 has length at most $|\psi|$, and $|I|\leq 2^{|\psi|}$.

 Obviously, $|\psi(\G)|=\sum_{i\in I} |\gamma_i(\G)|$.
 We now use, for each $i\in I$, the algorithm from
 Lemma~\ref{lemma:counting-graphs} to compute the number
 $s_i=|\gamma_i(\G)|$ and output the value $s=\sum_{i\in I}s_i$.

 By Lemma~\ref{lemma:counting-graphs} we know that for each $i\in I$
 computing  $s_i$ can be done in time
 $O(2^m{\cdot}|\gamma_i|{\cdot}\tilde{n}{\cdot}\tilde{d}^{h_0(|\gamma_i|)})$, where
 $m$ is the number of binary atoms in $\gamma$, $\tilde{n}=|\dom(\G)|$,  $\tilde{d}=\degree(\G)$, and
 $h_0$ is some computable function.

 By Proposition~\ref{prop-quantifier-elim-detailed} we know that
 $\tilde{n}\leq h(|\varphi|){\cdot}n{\cdot}d^{h(|\varphi|)}$ and
 $\tilde{d}\leq d^{h(|\varphi|)}$.
Since also $|\gamma_i|\leq |\psi|\leq h(|\varphi|)$,  the computation
of $s_i$, for each $i\in I$, takes time
 $h_1(|\varphi|){\cdot}n{\cdot}d^{h_1(|\varphi|)}$, for some computable
 function $h_1$ (depending on $h$ and $h_0$).

 To conclude, since $|I|\leq 2^{|\psi|}$, the total  running time for the computation of $|\varphi(\A)| =
  \sum_{i\in I}s_i$ is
  $h_2(|\varphi|){\cdot}n{\cdot}d^{h_2(|\varphi|)}$, for a suitably chosen computable
  function $h_2$. Hence, we meet the required bound.
\end{proof}

Theorem~\ref{thm:counting} is an immediate consequence of
Proposition~\ref{prop:counting}, following arguments similar with the proof of
Proposition~\ref{prop:connected-cqs_low-degree}: For a given $\epsilon>0$ we
let $\delta\deff \frac{\epsilon}{h(|\varphi|)}$, where $h$ is the function of
Proposition~\ref{prop:counting}, and define $n_\epsilon\deff m_\delta$.  Then,
every $\A\in\Class$ with $|\A|\geq n_\epsilon$ has
$\degree(\A)\leq |\A|^{\epsilon/h(|\varphi|)}$.  Thus, on input of $\A$ and $\varphi$, the
algorithm from Proposition~\ref{prop:counting} has running time
$O(h(|\varphi|){\cdot} |\A|^{1+\epsilon})$ if $|\A| \geq n_\epsilon$ and takes time
bounded by $h(|\varphi|)\cdot n_\epsilon^{1+h(|\varphi|)}$ otherwise. This gives the bounds
claimed by the proposition with a computable function $g$ as soon as we can
compute $n_\e$.

\subsection{Testing}\label{subsection:Testing}

Here we consider the problem of testing whether a given tuple is a solution to a query.
By Proposition~\ref{prop-quantifier-elim-detailed} it is enough to consider
quantifier-free formulas. Those are treated using the data structure computed
by Corollary~\ref{cor-test}.

\begin{prop}\label{prop:testing}
 There is an algorithm which at input of a structure $\A$, a
 first-order query $\varphi(\ov{x})$, and an $\e>0$ has a preprocessing phase of
 time  $g(|\varphi|,\e){\cdot}n^{1+\e}{\cdot}d^{g(|\varphi|,\e)}$
 in which it computes a data structure such that, on input of any
 $\ov{a}\in\dom(\A)^k$ for $k=|\ov{x}|$, it can be tested in time
 $g(|\varphi|,\e)$ whether $\ov{a}\in\varphi(\A)$,
 where
 $g$ is a computable function,
 $n=|\dom(\A)|$, and
 $d=\degree(\A)$.
\end{prop}

\begin{proof}
  Fix $\e >0$.

 We first use the algorithm of
 Proposition~\ref{prop-quantifier-elim-detailed}
 to compute the graph $\G$, the quantifier-free formula
 $\psi(\ov{x})$ and the data structure for function $f$.
 For some computable function $h$,
 all of this is done within time
  $h(|\varphi|){\cdot}n^{1+\e}{\cdot}d^{h(|\varphi|)}$,
 and furthermore, $|\psi|\leq h(|\varphi|)$
 and $\degree(\G)\leq d^{h(|\varphi|)}$.
 Note that $\size{\G}\leq
 h(|\varphi|){\cdot}n{\cdot}d^{h(|\varphi|)}$.
 By construction, we furthermore know for all $\ov{a}\in\dom(\A)^k$
 that
 $\ov a \in \varphi(\A) \iff f(\ov a)\in \psi(\G)$.

Recall from  Proposition~\ref{prop-quantifier-elim-detailed} that
$\psi(\ov{x})$ is a quantifier-free formula built from atoms of the
form $E(y,z)$ and $C(y)$ for unary relation symbols $C$.
Thus, checking whether a given tuple $\ov{v}\in\dom(\G)^k$ belongs to
$\psi(\G)$ can be done easily, provided that one can check whether unary
atoms $C(u)$ and binary atoms $E(u,u')$ hold in $\G$ for given nodes
$u,u'$ of $\G$.

Let $\tilde{n}=|\dom(\G)|$ and $\tilde{d}=\degree(\G)$. Recall that
$\tilde{n}\leq h(|\varphi|){\cdot}n{\cdot}d^{h(|\varphi|)}$ and $\tilde{d}\leq
d^{h(|\varphi|)}$. We apply Corollary~\ref{cor-test}
to $\G$ and $\e$.

This gives an extra preprocessing time of
$O({\tilde{d}}^{\tilde{r}}\cdot\tilde{n}^{1+\e})$, where $\tilde{r}$
only depends on $\tau$, i.e., $\tilde{r}$ is bounded by
$\tilde{h}(|\varphi|)$ for some computable function $\tilde{h}$.
Inserting the known bounds on $\tilde{r}$, $\tilde{n}$, and
$\tilde{d}$ shows that this extra preprocessing time is in
$O(g(|\varphi|,\e) \cdot n^{1+\e} \cdot d^{g(|\varphi|,\e)})$ for some
computable function $g$.

Finally, the testing algorithm works as follows. Given a tuple $\ov{a}\in\dom(\A)^k$, we
first construct $\ov{v}\deff f(\ov a)$ and then check whether $\ov{v}\in \psi(\G)$.
Building $\ov{v}\deff f(\ov a)$ can be done in time $O(k^2)$ (see
Proposition~\ref{prop-quantifier-elim-detailed}),
and by Corollary~\ref{cor-test} checking whether
$\ov{v}\in \psi(\G)$ depends only on  $\psi$, $\e$ and $\sigma$.
Hence, we meet the required bound for testing.
\end{proof}

Theorem~\ref{thm:testing} is an immediate consequence of
Proposition~\ref{prop:testing} using the usual argument: For a given
query $\phi$ and an
$\epsilon>0$ we let $\delta\deff \frac{\epsilon}{2g(|\varphi|,\e/2)}$, where $g$ is
the function of Proposition~\ref{prop:testing}, and define
$n_\epsilon\deff m_\delta$.  Then, every $\A\in\Class$ with
$|\A|\geq n_\epsilon$ has $\degree(\A)\leq |\A|^{\epsilon/2g(|\varphi|,\e/2)}$.
Thus, on input of $\A$, $\varphi$ and $\e/2$, the testing algorithm from
Proposition~\ref{prop:testing} has preprocessing time
$O(g(|\varphi|,\e/2){\cdot} |\A|^{1+\epsilon})$ if $|\A| \geq n_\epsilon$ and takes
time bounded by $g(|\varphi|,\e/2)\cdot n_\epsilon^{1+\e/2+g(|\varphi|,\e/2)}$
otherwise, and it has testing time $O(g(|\varphi|,\e/2))$. This
gives the bounds claimed by the theorem with a computable function as
soon as we can compute $n_\e$.

\subsection{Enumeration}\label{subsection:Enumeration}

Here we consider the problem of enumerating the solutions to a given query.  We
first illustrate the proof of Theorem~\ref{thm:enum} with our running example.

\begin{exa}\label{example-enum}
  Consider again the query $q$ of Example~\ref{example-def}.
  In order to enumerate $q$ with constant delay
  over a class of low degree we proceed as
  follows. During the preprocessing phase we precompute those blue nodes that
  contribute to the answer set, i.e.\ such that there is a red node not
  connected to it. This is doable in pseudo-linear time because our class has
  low degree and each blue node is connected to few red nodes. We call green
  the resulting nodes.  We then order the green nodes and the red nodes in
  order to be able to iterate through them with constant delay.  Finally, we
  compute the binary function $\fun{skip}(x,y)$ associating to each green node
  $x$ and red node $y$ such that $E(x,y)$ the smallest red node $y'$ such that
  $y<y'$ and $\lnot E(x,y')$, where $<$ is the order on red nodes precomputed
  above. From Proposition~\ref{prop:connected-cqs_low-degree} it follows that
  computing $\fun{skip}$ can be done in pseudo-linear time. It is crucial here
  that the domain of $\fun{skip}$ has pseudo-linear size and this is a
  consequence of the low degree.

  The enumeration phase now goes as follows: We iterate through all green
  nodes. For each of them we iterate through all red nodes. If there is no edge
  between them, we output the result and continue with the next red node. If
  there is an edge, we apply $\fun{skip}$ to this pair and the process
  continues with the resulting red node. Note that the new red node immediately
  yields an answer. Note also that all the red nodes that will not be considered are safely
  skipped as they are linked to the current green node.
\end{exa}

The proof of Theorem~\ref{thm:enum} can be sketched as follows.
By Proposition~\ref{prop-quantifier-elim-detailed} it is enough to consider
quantifier-free formulas looking for tuples of nodes that are disconnected and
have certain colors. Hence the query $q$ described in Example~\ref{example-def}
corresponds to the binary case. For queries of larger arities we proceed by
induction on the arity. If $q$ is given by the formula $\varphi(\bar x y)$ we
know by induction that, modulo some preprocessing pseudo-linear in the size of
the input database $D$, we can enumerate with constant delay all tuples
$\bar a$ satisfying $D \models \exists y \varphi(\bar x y)$. For each such
tuple $\bar a$ it remains to enumerate all $b$ such that
$D\models \varphi(\bar a b)$. We then proceed as in
Example~\ref{example-enum}. Starting from an arbitrary node $b$ of the
appropriate color, we iterate the following reasoning. If the current node $b$
is not connected to $\bar a$, then $\bar a b$ forms an answer and we proceed to
the next $b$. If $b$ is connected to $\bar a$ then we need to jump in constant
time to the next node of the appropriate color forming a solution. This is done
by precomputing a suitable function $\fun{skip}$ that depends on the arity of
the query and is slightly more complex that the one described in
Example~\ref{example-enum}. The design and computation of this function is the
main technical originality of the proof. The fact that the database has low
degree implies that for each tuple $\bar a$ there are few nodes $b$ that are
connected to $\bar a$. This makes the computation efficient.

The technical details are summarized in the following proposition.

\begin{prop}\label{prop:enumeration}
 There is an algorithm which at input of a structure $\A$, a
 first-order query $\varphi(\ov{x})$, and $\e>0$ enumerates $\varphi(\A)$ with
 delay $h(|\varphi|,\e)$ after a preprocessing of time
 $h(|\varphi|,\e){\cdot}n^{1+\e}{\cdot}d^{h(|\varphi|,\e)}$, where
 $n=|\dom(\A)|$, $d=\degree(\A)$, and  $h$ is a computable function.
\end{prop}
\begin{proof}
  The proof is by induction on the number $k\deff |\ov{x}|$ of free variables of $\varphi$.
 In case that $k=0$, the formula $\varphi$ is a sentence, and we are done using
 Theorem~\ref{thm:grohe-low-degree}.
 In case that $k>0$ we proceed as follows.

 We first use the algorithm of Proposition~\ref{prop-quantifier-elim-detailed}
 to compute the according colored graph $\G$, the quantifier-free formula
 $\psi(\ov{x})$, and the data structure representing $f$.  This takes time
 $g(|\varphi|){\cdot}n^{1+\e}{\cdot}d^{g(|\varphi|)}$ for a computable function
 $g$.  And we know that $|\psi|\leq g(|\varphi|)$, that $\G$ has degree
 $\tilde{d}\leq d^{g(|\varphi|)}$, and that $\dom(\G)$ has $\tilde{n}$
 elements, where $\tilde{n}\leq g(|\varphi|){\cdot}n{\cdot}d^{g(|\varphi|)}$.

 From Item~\ref{item:psi-form:prop-quantifier-elim-detailed} of
 Proposition~\ref{prop-quantifier-elim-detailed} we know that the formula
 $\psi(\ov{x})$ is of the form\,$(\psi_{1} \und \psi_{2})$, where $\psi_{1}$
 states that no distinct free variables of $\psi$ are connected by an $E$-edge
 and $\psi_{2}$ is a positive boolean combination of unary atoms.

 We now prove the Proposition in the case of $\G$, i.e.\ we enumerate
 $\psi(\G)$ and go back to $\A$ using $f$ in constant
 time by Item~\ref{item:f:prop-quantifier-elim-detailed} of
 Proposition~\ref{prop-quantifier-elim-detailed}. We assume an arbitrary linear
 order $\leq_\G$ among the nodes of $\G$.

 In case that $k=1$, $\psi(x_1)=\psi_2(x_1)$ is a positive boolean combination
 of unary atoms. We can use Lemma~\ref{lemma:connected-cqs} for each unary atom
 in order to compute $\psi(\G)$ in time
 $O(|\psi|{\cdot}\tilde{n}{\cdot}\tilde{d}^{g(|\psi|)})$ for a computable
 function $g$. From this the constant delay enumeration is immediate.

 Altogether the preprocessing time is in
 $h(\varphi){\cdot}n^{1+\e}{\cdot}d^{h(\varphi)}$, for a computable function $h$ as
 required. The delay is $O(1)$, and we  are done for $k=1$.

 The case $k>1$ requires much more elaborate constructions.

 We let $\ov{x}=(x_1,\ldots,x_k)$ and $\ov{x}_{k-1}\deff (x_1,\ldots,x_{k-1})$.
 We first transform $\psi$ into a normal form $\Oder_{j\in J}\theta_j(\ov{x})$
 such that the formulas $\theta_j$ exclude each other (i.e., for each
 $\ov{v}\in\psi(\G)$ there is exactly one $j\in J$ such that
 $\ov{v}\in\theta_j(\G)$), and each $\theta_j(\ov{x})$ is of the form
 \[
  \begin{array}{c}
    \phi_j(\ov{x}_{k-1}) \ \und \
     P_j(x_k) \ \und\ \gamma(\ov{x}), \qquad\text{where}
\\[2ex] \displaystyle
  \gamma(\ov{x}) \ \ \deff \ \
  \Und_{i=1}^{k-1} \big(\,\nicht E(x_i,x_k)\, \und\, \nicht
  E(x_k,x_i)\,\big),
\end{array}
\]
$P_j(x_k)$ is a boolean combination of unary atoms regarding $x_k$, and
$\phi_j(\ov{x}_{k-1})$ is a formula with only $k{-}1$ free variables.
Note that the transformation into this normal form can be done
easily, using the particularly simple form of the formula $\psi$.

As the $\theta_j$ are mutually exclusive, we can enumerate $\psi(\G)$ by
enumerating for each $j\in J$, $\theta_j(\G)$.

In the following, we therefore restrict attention to $\theta_j$ for a fixed $j\in J$. For
this $\theta_j$ we shortly write
\[
  \theta(\ov{x}) \ = \ \
    \phi(\ov{x}_{k-1}) \ \und \
     P(x_k) \ \und\ \gamma(\ov{x}).
\]
We let
\ \(
    \theta'(\ov{x}_{k-1}) \ \deff \
    \exists x_k\, \theta(\ov{x}).
 \) \

 By induction hypothesis (since $\theta'$ only has $k{-}1$ free variables) we
 can enumerate $\theta'(\G)$ with delay $h(|\theta'|,\e)$ and preprocessing
 $h(|\theta'|,\e){\cdot}{\tilde n}^{1+\e}{\cdot}\tilde{d}^{h(|\theta'|,\e)}$

Since $P(x_k)$ is a boolean combination of unary atoms on $x_k$, we
can use Lemma~\ref{lemma:connected-cqs} to compute $P(\G)$
 in time $O(|P|{\cdot}\tilde{n}{\cdot}\tilde{d}^{g(|P|)})$ for a
 computable function $g$.
Afterwards, we have available a list of all nodes $v$ of $\G$ that
belong to $P(\G)$. In the following, we will write $\leq_\G^P$ to denote the
linear ordering of $P(\G)$ induced by $\leq_\G$ on $P(\G)$, and we write
$\first_\G^P$ for the first element in this list, and $\next^P_\G$ for the
successor function, such that for any node $v\in P(\G)$,
$\next_\G^P(v)$ is the next node in $P(\G)$ in this list (or the value
$\void$, if $v$ is the last node in the list).

We extend the signature of $\G$ by a unary relation symbol $P$ and a
binary relation symbol $\next$, and let $\hat{\G}$ be the expansion of
$\G$ where $P$ is interpreted by the set $P(\G)$ and $\next$ is
interpreted by the successor function $\next_\G^P$ (i.e., $\next(v,v')$
is true in $\hat{\G}$ iff $v'=\next_\G^P(v)$). Note that $\hat{\G}$ has
degree at most $\hat{d}=\tilde{d}{+}2$.

We now start the key idea of the proof, i.e., the function that will help us
skipping over irrelevant nodes. To this end consider the first-order
formulas $E_1,\ldots,E_k$ defined inductively as follows, where $E'(x,y)$ is an
abbreviation for $(E(x,y)\oder E(y,x))$. The reason for defining these formulas
will become clear only later on, in the proof.
\[
  E_1(u,y) \ \deff \ \ E'(u,y), \qquad\text{and}
\]
\[
 \begin{array}{l}
  E_{i+1}(u,y) \ \deff \ E_i(u,y) \ \ \oder \ \
  \exists z\exists z'\exists v \, \big(E'(z,u) \und \next(z',z) \und E'(v,z') \und E_i(v,y)
  \big).
 \end{array}
\]

 \noindent
 A simple induction shows that for $E_i(u,y)$ to hold, $y$
 must be at distance $\leq 3(i{-}1)+1 < 3i$ from $u$.

 In our algorithm we will have to test, given $i\leq k$ and nodes $u,v\in\dom(\G)$,
 whether $(u,v)\in E_i(\hat{\G})$.
 Since $E_i$ is a first-order formula,
 Proposition~\ref{prop:testing} implies that,
 after a preprocessing phase using time
 $g'(|E_i|,\e){\cdot}\tilde{n}^{1+\e}{\cdot}\hat{d}^{g'(|E_i|,\e)}$ (for some computable
 function $g'$), testing membership in $E_i(\hat{\G})$, for any given $(u,v)\in\dom({\G})^2$,
is possible within time $g'(|E_i|,\e)$.

The last step of the precomputation phase computes the function
$\fun{skip}$ that associates to each node $y\in P(\G)$ and each set $V$
of at most $k{-}1$ nodes that are related to $y$ via $E_k$, the smallest (according
to the order $\leq_\G^P$ of $P(\G)$) element $z\geq_\G^P y$ in $P(\G)$ that is
\emph{not} connected by an $E$-edge to any node in $V$. More precisely:
For any node $y\in P(\G)$ and any set $V$ with $0\leq |V| < k$ and
$(v,y)\in E_k(\hat{\G})$ for all $v\in V$, we let
\begin{equation*}
\fun{skip}(y,V)\ := \ \
\min\{z \in P(\G) \ : \  y \,{\leq_\G^P} z \ \text{ and } \
 \forall v\in
  V:\  (v,z)\not\in E'(\G)\},
\end{equation*}
respectively, $\fun{skip}(y,V)\deff\void$ if no such $z$ exists.

Notice that the nodes of $V$ are related to $y$ via $E_k$ and hence are at
distance $<3k$ from $y$.
Hence for each $y$, we only need to consider at most
${\hat{d}}^{(3k^2)}$ such sets $V$.

For each set $V$, $\fun{skip}(y,V)$ can be computed  by running
consecutively through all nodes $z\geq_\G^P y$ in the list $P(\G)$ and test
whether $E'(z,v)$ holds for some $v\in V$. This can be done in constant time as
we have done the preprocessing for testing for all of the $E_i$.

Since $|V|\leq k$ and each
$v\in V$ is of degree at most $\tilde{d}$ in $\G$, the value  $\fun{skip}(y,V)$  can be
found in time $O(k^2{\cdot} \tilde{d})$. Therefore, the entire
$\fun{skip}$-function can be computed, and stored in a data structure by Theorem~\ref{thm-storing-complete}, in time
$O(\tilde{n}^{1+\e} {\cdot} {\hat{d}}^{(3k^2)}{\cdot}
g''(|\varphi|,\e))$ for some computable function $g''$. Later, given $y$ and $V$ as above, the value
$\fun{skip}(y,V)$ can be looked-up within constant time.

We are now done with the preprocessing phase.
%
 Altogether it took

 \begin{enumerate}
 \item
  the time to compute $\psi$ and $\G$, which is
  $g(|\varphi|){\cdot}n^{1+\e}{\cdot}d^{g(|\varphi|)}$, for a computable
  function $g$

 \item
  the time to compute $\Oder_{j\in J}\theta_j$, which is
  $g(|\varphi|)$, for a computable function $g$

 \item
  for each $j\in J$ and $\theta\deff \theta_j$, it took

 \begin{enumerate}

 \item
  the preprocessing time for $\theta'(\G)$,\\
  which is by induction
  $h(|\theta'|,\e){\cdot}\tilde{n}^{1+\e}{\cdot}\tilde{d}^{h(|\theta'|,\e)}$,
  for the computable function $h$ in the Proposition's statement

 \item
  the time for computing $P(\G)$, which is\\
  $g(|\varphi|){\cdot}\tilde{n}{\cdot}\tilde{d}^{g(|\varphi|)}$, for a computable
  function $g$

\item for all $i\leq k$, the preprocessing time for testing membership in
  $E_i(\hat{\G})$, which can be done in time
  $g'(|E_i|,\e){\cdot}\tilde{n}^{1+\e}{\cdot}\hat{d}^{g'(|E_i|,\e)}$, for a computable function $g'$

 \item and
  the time for computing the $\fun{skip}$-function and to compute the
  associated data structure, which is $g''(|\varphi|,\e){\cdot}\tilde{n}^{1+\e}{\cdot} \tilde{d}^{g''(|\varphi|,\e)}$, for a
  computable function $g''$.

 \end{enumerate}
\end{enumerate}

\noindent
It is straightforward to see that, by suitably choosing the computable
function $h$, all the preprocessing steps can be done within time
$h(|\varphi|,\e){\cdot}n^{1+\e}{\cdot}d^{h(|\varphi|,\e)}$.

\medskip

We now turn to the enumeration procedure.
As expected we enumerate all tuples $\bar u \in \psi(\G)$ and return
$f^{-1}(\bar u)$.

In order to enumerate $\psi(\G)$ it suffices to enumerate $\theta_j(\G)$ for
each $j$. Fix $j$ and let $\theta=\theta_j$ and $\theta'$ be as in the
preprocessing phase. The enumeration of $\theta(\G)$ is done as follows.

\begin{enumerate}
  \item[1.]
    Let $\ov{u}$ be the first output produced in
    the enumeration of $\theta'(\G)$.\\
    If $\ov{u}=\void$ then STOP with output $\void$,\\
    else let
    $(u_1,\ldots,u_{k-1})=\ov{u}$ and goto line 2.
  \item[2.]
    Let $y\deff \first^P_\G$ be the first element in the list $P(\G)$.
  \item[3.]
    Let $V\deff\setc{v\in\set{u_1,\ldots,u_{k-1}}}{(v,y)\in E_k(\hat{\G})}$.
  \item[4.]
    Let $z\deff\fun{skip}(y,V)$.
  \item[5.]
    If $z\neq\void$ then OUTPUT $(\ov{u},z)$
    and goto line~9.
  \item[6.]
    If $z= \void$ then
   \begin{enumerate}
     \item[7.]
      Let $\ov{u}'$ be the next output produced in the enumeration
      of $\theta'(\G)$.
     \item[8.]
      If $\ov{u}'=\void$ then STOP with output $\void$, \\
      else let $\ov{u}\deff \ov{u}'$ and goto line 2.
   \end{enumerate}
  \item[9.]
    Let $y\deff \next^P_\G(z)$.
  \item[10.]
    If $y=\void$ then goto line 7, else goto line 3.
\end{enumerate}

\noindent
We prove that the above process enumerates $\theta(\G)$ with constant delay.

To see this, notice first that the algorithm never outputs any tuple more than once.
Before proving that this algorithm enumerates exactly the tuples in
$\theta(\G)$, let us first show that it operates with delay at most $h(|\varphi|,\e)$.

By the induction hypothesis, the execution of line~1 and
 each execution of line~7 takes time at most $h(|\theta'|,\e)$.
Furthermore, each execution of line~3 takes time
$(k{-}1){\cdot}g'(|E_k|,\e)$. Concerning the remaining lines of the
algorithm, each execution can be done in time $O(1)$.

Furthermore, before outputting the first tuple, the algorithm executes
at most 5 lines (namely, lines 1--5; note that by our choice of the
formula $\theta'$ we know that when entering line~5 before outputting
the first tuple, it is guaranteed that $z\neq\void$, hence an output
tuple is generated).

Between outputting two consecutive tuples, the algorithm
executes at most 12 lines (the worst case is an execution of lines 9, 10,
3, 4, 5, 6, 7, 8, 2, 3, 4, 5; again, by our choice of the
formula $\theta'$, at the last execution of line~5 it is guaranteed
that $z\neq \void$, hence an output tuple is generated).

Therefore, by suitably choosing the function $h$, we obtain that the
algorithm enumerates with delay at most $h(|\varphi|,\e)$.

\medskip
We now show that any tuple outputed is a solution.
To see this consider a tuple $\bar u z$ outputed at step 5. By construction we
have $z\in P(\hat{\G})$ and $\bar u \in \phi(\hat{\G})$. Hence in order to show
that $\bar uz \in \theta(\hat{\G})$, it remains to
verify that $z$ is not connected to any of the elements in $\bar u$.

By definition of $\fun{skip}$ it is clear that $z$ is not connected to the elements
of $V$. Assume now that $z=y$. Then by definition of $V$, $z$ is also not connected
to all the elements not in $V$ and we are done.

We can therefore assume that $z>^P_\G y$. Let $z'$ be the predecessor of $z$ in
$P$ (i.e. $\next(z')=z$, possibly $z'=y$). Assume towards a contradiction that
$z$ is connected to an element $x$ of $\bar u$. From the remark above, we know
that $x \not\in V$. As $z'$ was skipped by $\fun{skip}$ this means that $(z',v)
\in E'({\G})$ for some $v\in V$. Consider $c \in V$. By definition of $V$ we have
$(c,y)\in E_k(\hat{\G})$. It turns out that $(c,y) \in E_{k-1}(\hat{\G})$. This
is because when one $E_j$ does not produce anything outside of $E_{j-1}$ then
all the $E_{j'}$ for $j'>j$ also do not produce anything outside of
$E_{j-1}$. Hence, as $|V|<k$, we must have $(c,y) \in E_{k-1}(\hat{\G})$ and in
particular $(v,y) \in E_{k-1}(\hat{\G})$. Altogether, $z,z',v$ witness the fact that
$(x,y)\in E_k(\hat{\G})$ contradicting the fact that $x$ is not in $V$.

\medskip

It remains to show that all tuples are outputed. {This is done by
  induction.  By induction we know that we consider all relevant $\bar u$. By
  definition, the function $\fun{skip}$ skips only elements $y$ such that
  $\bar u y$ is not a solution. Therefore we eventually output all solutions.}
\end{proof}

Theorem~\ref{thm:enum} follows immediately from Proposition~\ref{prop:enumeration}
using the usual argument: For a given
$\epsilon>0$ we let $\delta\deff \frac{\epsilon}{2h(|\varphi|,\e/2)}$, where $h$ is
the function of Proposition~\ref{prop:enumeration}, and define
$n_\epsilon\deff m_\delta$.  Then, every $\A\in\Class$ with
$|\A|\geq n_\epsilon$ has $\degree(\A)\leq |\A|^{\epsilon/2h(|\varphi|,\e)}$.
Thus, on input of $\A$, $\varphi$ and $\e/2$, the enumeration algorithm from
Proposition~\ref{prop:testing} has preprocessing time
$O(h(|\varphi|,\e/2){\cdot} |\A|^{1+\epsilon})$ if $|\A| \geq n_\epsilon$ and takes
time bounded by $h(|\varphi|,\e/2)\cdot n_\epsilon^{1+\e/2+h(|\varphi|,\e/2)}$
otherwise and delay time $h(|\varphi|,\e)$. This
gives the bounds claimed by the proposition with a computable function as
soon as we can compute $n_\e$.

\section{Proof of  quantifier elimination and normal form}\label{section-main-proofs}


This section is devoted to the proof of Proposition~\ref{prop-quantifier-elim-detailed}.
The proof consists of several steps, the first of which relies on a
transformation of $\varphi(\ov{x})$ into an equivalent formula in
Gaifman normal form, i.e., a boolean combination of basic-local sentences
and formulas that are local around $\ov{x}$.
A formula $\lambda(\ov{x})$ is \emph{$r$-local}
around $\ov{x}$ (for some $r\geq 0$) if every quantifier is
relativized to the $r$-neighborhood of $\ov{x}$.
A \emph{basic-local sentence} is
of the form
\[
  \exists y_1\cdots \exists y_\ell \!\!\Und_{1\leq i<j\leq
  \ell}\!\!\! \dist(y_i,y_j)> 2r \ \ \und \ \Und_{i=1}^\ell \theta(y_i),
\]
where
$\theta(y)$ is $r$-local around $y$.
By Gaifman's
well-known theorem we obtain an algorithm that transforms an input
formula $\varphi(\ov{x})$ into an equivalent formula in Gaifman normal
form~\cite{Gaifman-82}.

The rest of the proof can be sketched as follows. Basic-local sentences can be
evaluated on classes of structures of low degree in pseudo-linear time by Theorem~\ref{thm:grohe-low-degree}, so it
remains to treat formulas that are local around their free
variables. By the
Feferman-Vaught Theorem (cf., e.g.~\cite{Makowsky04}), we can further decompose local formulas into formulas that are
local around \emph{one} of their free variables. The latter turns out to have a small
answer set that can be precomputed in pseudo-linear time. The remaining
time is used to compute the structures useful for reconstructing
the initial answers from their components. We now 
give the details.

\begin{proof}[Proof of Proposition~\ref{prop-quantifier-elim-detailed}]
\ \\
\noindent
\emph{Step~1: transform $\varphi(\ov{x})$ into a local formula $\varphi'(\ov{x})$.}

We first transform $\varphi(\ov{x})$ into an equivalent formula
$\varphi^G(\ov{x})$ in Gaifman normal form.
%
For each basic-local sentence $\chi$ occurring in
$\varphi^G(\ov{x})$, check whether
$\A\models\chi$ and let $\chi'\deff \true$ if $\A\models\chi$ and
$\chi'\deff \false$ if $\A\not\models\chi$.
Let $\varphi'(\ov{x})$ be the formula obtained from
$\varphi^G(\ov{x})$ by replacing every
basic-local sentence $\chi$ occurring in $\varphi^G(\ov{x})$ with
$\chi'$.
By using Gaifman's theorem and
Theorem~\ref{thm:grohe-low-degree},
all this can be done in time  $O(h(|\varphi|){\cdot}
n{\cdot} d^{{h}(|\varphi|)})$, for a computable function $h$.

Clearly,
for every $\ov{a}\in \dom(\A)^k$ we have
$\A\models\varphi'(\ov{a})$
iff
$\A\models\varphi(\ov{a})$.
Note that there is a number $r\geq 0$ such that $\varphi'(\ov{x})$ is
$r$-local around $\ov{x}$, and this number
is provided as a part of the output of Gaifman's algorithm.

\medskip

\noindent
\emph{Step~2: transform $\varphi'(\ov{x})$ into a disjunction
  $\Oder_{P\in \Part}\psi'_P(\ov{x})$.}

Let $\ov{x}=(x_1,\ldots,x_k)$.
A \emph{partition} of the set $\set{1,\ldots,k}$ is a list
$P=(P_1,\ldots,P_\ell)$ with $1\leq \ell\leq k$ such that

\begin{mi}
    \item
      $\emptyset\neq P_j \subseteq \set{1,\ldots,k}$, for every
      $j\in\set{1,\ldots,\ell}$,

  \item
      $P_1\cup\cdots\cup P_\ell=\set{1,\ldots,k}$,

   \item
      $P_j\cap P_{j'}=\emptyset$, for all $j,j'\in\set{1,\ldots,\ell}$
      with $j\neq j'$,

   \item
      $\min P_j < \min P_{j+1}$, for all $j\in\set{1,\ldots,\ell{-}1}$.
\end{mi}
Let $\Part$ be the set of all partitions of
$\set{1,\ldots,k}$. Clearly, $|\Part|\leq k!$. 
   For each $P=(P_1,\ldots,P_\ell)\in\Part$ and each
   $j\leq \ell$
   let $\ov{x}_{P_j}$ be the tuple obtained
   from $\ov{x}$ by deleting all those $x_i$ with $i\not\in P_j$.

   For every partition $P=(P_1,\ldots,P_\ell)\in\Part$ let
   $\rho_P(\ov{x})$ be an $\FO(\sigma)$-formula stating that each
   of the following is true:

   \begin{enumerate}
    \item
      The $r$-neighborhood around $\ov{x}$ in $\A$
      is the disjoint union of the
      $r$-neighborhoods around $\ov{x}_{P_j}$
      for
      $j\leq \ell$.
      I.e.,
   \[
\delta_{P}(\ov{x}) \ \deff \
      \Und_{1\leq j<j'\leq \ell}\  \Und_{(i,i')\in P_j\times P_{j'}} \dist(x_i,x_{i'})>2r{+}1.
   \]

    \item
      For each
      $j\leq \ell$,
      the $r$-neighborhood around
      $\ov{x}_{P_j}$ in $\A$ is connected, i.e., satisfies the formula
   \[
\gamma_{P_j}(\ov{x}_{P_j}) \ \deff \
      \Oder_{\substack{E \subseteq P_j{\times} P_j
      \text{ such that the} \\ \text{graph }(P_j,E) \text{ is connected}}} \ \Und_{(i,i')\in E}
    \dist(x_i,x_{i'})\leq 2r{+}1.
   \]
   \end{enumerate}
\noindent
   Note that the formula
   \ \(
     \rho_P(\ov{x})  \ \deff \
     \delta_{P}(\ov{x}) \und \Und_{j=1}^\ell\gamma_{P_j}(\ov{x}_{P_j})
   \) \
   is $r$-local around $\ov{x}$.
 Furthermore,
 $\varphi'(\ov{x})$ obviously is equivalent to the formula
   \ \(
    \Oder_{P\in\Part} \big( \,
      \rho_P(\ov{x}) \und \varphi'(\ov{x})
    \, \big).
   \)

  Using the Feferman-Vaught Theorem (see e.g.~\cite{Makowsky04}),
  we can, for each

  $P=(P_1,\ldots,P_\ell)\in\Part$, compute a decomposition of
  $\varphi'(\ov{x})$ into $r$-local formulas $\vartheta_{P,j,t}(\ov{x}_{P_j})$,
  for $j\in\set{1,\ldots,\ell}$ and $t\in T_P$, for a suitable finite
  set $T_P$,
  such that the formula $\big(\rho_P(\ov{x})\und\varphi'(\ov{x})\big)$
  is equivalent to
  \begin{equation*}\label{eq:FefermanVaught-formula}
    \rho_P(\ov{x}) \ \und \ \Oder_{t\in T_P}\ \big( \
      \vartheta_{P,1,t}(\ov{x}_{P_1}) \ \und \ \cdots \ \und \
      \vartheta_{P,\ell,t}(\ov{x}_{P_\ell})
     \ \big)
  \end{equation*}
  which, in turn, is equivalent to
  \,$\psi'_{P} \deff (\psi'_{P,1} \und \psi'_{P,2})$, where
  $\psi'_{P,1}\deff \delta_P(\ov{x})$ and

  \[
   \psi'_{P,2} \ \ \deff\ \
   \Big(\,\Und_{j=1}^\ell\gamma_{P_j}(\ov{x}_{P_j})\,\Big) \ \, \und \ \,
   \Oder_{t\in T_P} \Big(\, \Und_{j=1}^\ell
      \vartheta_{P,j,t}(\ov{x}_{P_j})\,\Big).
  \]

  In summary, $\varphi'(\ov{x})$ is equivalent to
  $
    \Oder_{P\in\Part} \psi'_{P}(\ov{x}),
  $
  and for every tuple $\ov{a}\in \dom(\A)^k$ with
  $\A\models\varphi'(\ov{a})$, there is exactly one partition $P\in\Part$
  such that  $\A\models\psi'_{P}(\ov{a})$ (since $\A \models \rho_P(\ov{a})$ is true for only one such $P\in \Part$).

 \medskip

 \noindent
 \emph{Step~3: defining $\G$, $f$, and $\psi$.}

 We define the domain $G$ of $\G$ to be the disjoint union of the sets
 $A$ and $V$, where $A\deff\dom(\A)$, and $V$
 consists of a ``dummy element'' $\dummy$, and
an element  $v_{(\ov{b},\iota)}$
\begin{itemize}
 \item
  for each  $\ov{b}\in A^1\cup\cdots\cup A^k$ such that
  $\A\models \gamma_{P}(\ov{b})$ where $P\deff\set{1,\ldots,|b|}$
  and
 \item
  for each injective mapping
  $\iota:\set{1,\ldots,|\ov{b}|}\to\set{1,\ldots,k}$.
\end{itemize}
 Note that the first item ensures that the $r$-neighborhood around
 $\ov{b}$ in $\A$ is connected. The second item ensures that we can
 view $\iota$ as a description telling us that the $i$-th component of
 $\ov{b}$ shall be viewed as an assignment for the variable
 $x_{\iota(i)}$ (for each $i\in\set{1,\ldots,|\ov{b}|}$).

 We let $f$ be the function from $A^k$ to $V^k$ defined as follows:
 For each $\ov{a}\in A^k$ let
 $P=(P_1,\ldots,P_\ell)$ be the unique element in $\Part$ such that
 $\A\models\rho_P(\ov{a})$.
 For each $j\leq \set{1,\ldots,\ell}$, we write $\ov{a}_{P_j}$ for
 the tuple obtained from $\ov{a}$ by deleting all those $a_i$ with
 $i\not\in P_j$.
 Furthermore, we
 let $\iota_{P_j}$ be the
 mapping from $\set{1,\ldots,|P_j|}$ to $\set{1,\ldots,k}$
 such that $\iota(i)$ is the
 $i$-th smallest element of $P_j$, for any $i\in\set{1,\ldots,|P_j|}$.
 Then,
 \[
   f(\ov{a}) \ \deff \
   \big(\, v_{(\ov{a}_{P_1},\iota_{P_1})},\ldots,v_{(\ov{a}_{P_\ell},\iota_{P_\ell})},
   \dummy, \ldots, \dummy \, \big),
 \]
 where the number of $\dummy$-components is $(k{-}\ell)$.
 It is straightforward to see that $f$ is injective.

 We let $\tau_1$ be the signature consisting of a unary relation
 symbol $C_\bot$, and a unary relation symbol $C_\iota$ for each injective
 mapping $\iota:\set{1,\ldots,s}\to\set{1,\ldots,k}$ for
 $s\in\set{1,\ldots,k}$.

 In $\G$, the symbol $C_\bot$ is interpreted by the singleton set
 $\set{\dummy}$, and each
 $C_\iota$ is interpreted by the set of all nodes
 $v_{(\ov{b},\hat{\iota})}\in V$ with $\hat{\iota}=\iota$.

 We let $E$ be a binary relation symbol which
 is interpreted in $\G$
 by the set of all tuples $(v_{(\ov{b},\iota)},v_{(\ov{c},\hat{\iota})})\in V^2$ such that
 there are elements $b'\in A$ in $\ov{b}$ and $c'\in A$ in $\ov{c}$
 such that $\dist^\A(b',c')\leq 2r{+}1$.

 For each $P=(P_1,\ldots,P_\ell)\in\Part$,
 each $j\in\set{1,\ldots,\ell}$, and
 each $t\in T_P$ we let $C_{P,j,t}$ be a unary
 relation symbol which, in $\G$, is interpreted by the set of all
 nodes $v_{(\ov{b},\iota)}\in V$ such that $\iota=\iota_{P_j}$ and
 $\A\models\vartheta_{P,j,t}(\ov{b})$.

 We let $\tau_2$ be the signature consisting of
 all the unary relation symbols $C_{P,j,t}$.

 We let $\ov{y}=(y_1,\ldots,y_k)$ be a tuple of $k$ distinct variables,
 and we define $\psi_1(\ov{y})$ to be the $\FO(E)$-formula
 \[
   \psi_{1}(\ov{y}) \ \ \deff \ \
   \Und_{\substack{ 1\leq j,j'\leq k \\ \text{with }j\neq j'}}\!\!\!\!\!\nicht E(y_j,y_{j'}).
 \]
 For each $P=(P_1,\ldots,P_\ell)$ we let
 $\psi_P(\ov{y})$ be the $\FO(\tau_1\cup\tau_2)$-formula
 defined as follows:
 \begin{eqnarray*}
   \psi_{P}(\ov{y}) & \deff &
     \Big(\Und_{j=1}^\ell C_{\iota_{P_j}}(y_j) \Big) \ \, \und \ \,
     \Big(\Und_{j=\ell+1}^k C_\bot(y_j) \Big) \ \, \und \, \
     \Oder_{t\in T_P}\Big( \Und_{j=1}^\ell C_{P,j,t}(y_j)\Big).
 \end{eqnarray*}
 It is straightforward to verify that the following is true:
 \begin{me}
  \item
   For every $\ov{a}\in A^k$ with $\A\models \psi'_P(\ov{a})$,
   we have \ $\G\models (\psi_1\und\psi_P)(f(\ov{a}))$.
  \item
   For every $\ov{v}\in G^k$ with $\G\models
   (\psi_1\und\psi_P)(\ov{v})$,
   there is a (unique) tuple $\ov{a}\in A^k$ with $\ov{v}=f(\ov{a})$,
   and for this tuple we have $\A\models\psi'_P(\ov{a})$.
 \end{me}
 Finally, we let
 \[
   \psi(\ov{y}) \ \deff  \
     \big(\psi_1(\ov{y}) \und \psi_2(\ov{y})\big)
   \quad
   \text{with}
   \quad
   \psi_2(\ov{y}) \ \deff \
     \Oder_{P\in\Part} \psi_P(\ov{y}).
 \]
 It is straightforward to see that $f$ is a bijection between
 $\varphi(\A)$ and $\psi(\G)$.

 In summary, we now know that items~\ref{item:psi-form:prop-quantifier-elim-detailed}
 and~\ref{item:psi-comput:prop-quantifier-elim-detailed}, as well as the non
 computational part of item~\ref{item:f:prop-quantifier-elim-detailed} of
 Proposition~\ref{prop-quantifier-elim-detailed} are true.

 Later on in \emph{Step~5} we will provide details on how to build a
 data structure
 that, upon input of any tuple $\ov{a}\in A^k$, returns the tuple
 $f(\ov{a})$ within the claimed time bounds.

 In order to be able to also compute $f^{-1}(\ov{v})$ upon input of
 any tuple $\ov{v}\in \psi(\G)$,
 we use additional
 binary relation symbols $F_1,\ldots,F_k$ which are interpreted in
 $\G$ as follows:
 Start by initializing all of them to the empty set.
 Then, for each $v=v_{(\ov{b},\iota)}\in V$ and each
 $j\in\set{1,\ldots,|\ov{b}|}$, add to $F_{\iota(j)}^\G$ the tuple
 $(v,a)$, where $a$ is $j$-th component of $\ov{b}$.
 This completes the definition of $\G$ and $\tau$, letting
 $\tau\deff\tau_1\cup\tau_2\cup\set{E,F_1,\ldots,F_k}$.

 Using the relations $F_1,\ldots,F_k$ of $\G$, in time
 $O(k)$ we can, upon input of $v=v_{(\ov{b},\iota)}\in V$ compute the
 tuple $\ov{b}$ and the mapping $\iota$ (for this, just check for all
 $i\in\set{1,\ldots,k}$ whether node $v$ has an outgoing $F_i$-edge).
 Using this, it is straightforward to see
 that upon input of $\ov{v}\in \psi(\G)$, the tuple
 $f^{-1}(\ov{v})\in A^k$ can be computed in time  $O(k^2)$,

 \medskip

 \noindent
 \emph{Step~4: Computing \G within the time bounds of Item~\ref{item:G:prop-quantifier-elim-detailed}.}

 First of all,
 note that for each $v_{(\ov{b},\iota)}\in V$, the tuple $\ov{b}$ is
 of the form $(b_1,\ldots,b_s)\in A^s$ for some $s\leq k$, such that all
 components of the tuple belong to the $\hat{r}$-neighborhood
 $\N^{\A}_{\hat{r}}(b_1)$ of $b_1$ in
 $\A$, for $\hat{r}\deff k(2r{+}1)$.

By Lemma~\ref{lemma-compute-neigh} we can compute in total time $O(|\varphi|\cdot n
  \cdot d^{h'(|\varphi|)})$ all the neigborhoods
  $\N^{\A}_{\hat{r}}(a)$ (for all $a\in\dom(\A)$), where $h'$
is some computable function (recall that $\hat r$ depends only on
$\varphi$).
Within the same time bound, we can also compute all the neighborhoods
$\N^{\A}_{\tilde{r}}(a)$,
$\N^{\A}_{r'}(a)$, and
$\N^{\A}_{r'-r}(a)$,
for $\tilde{r}\deff \hat{r}{+}r$ and
 $r'\deff r+(2k{+}1)(2r{+}1)$
(later on it will be convenient to have efficient
access to all these neighborhoods).

Thus, the set $V$, along with the relations $C_\bot$, $C_\iota$ and $F_1,\allowbreak\ldots,F_k$ of $\G$,
 can be computed as follows: Start by letting $V\deff \set{\dummy}$
 and initializing all relations to the empty set. Let $C_{\bot}^{\G}\deff\set{\dummy}$.
 Then, for each $a\in A$,
 consider the $\hat{r}$-neighborhood
 $\N^{\A}_{\hat{r}}(a)$ of $a$ in
 $\A$, and compute (by a brute-force algorithm), for each
 $s\in\set{1,\ldots,k}$, the set of all $s$-tuples $\ov{b}$ of
 elements from this neighborhood, which satisfy the following:
 The first component of $\ov{b}$ is $a$, and
 $\N^{\A}_{\hat{r}}(a)\models\gamma_{P_j}(\ov{b})$ for
 $P_j=\set{1,\ldots,s}$.
 For each such tuple $\ov{b}$ do the following:
 For each injective mapping $\iota:\set{1,\ldots,s}\to\set{1,\ldots,k}$
 add to $V$ a new element $v_{(\ov{b},\iota)}$, add this element to
 the relation $C_{\iota}^{\G}$,
 and for each $j\in\set{1,\ldots,s}$, add to $F_{\iota(j)}^{\G}$ the tuple
 $(v_{(\ov{b},\iota)},a)$, where $a$ is the $j$-th component of $\ov{b}$.

 This way, the domain $G=A\cup V$ of $\G$, along with the relations
 $C_\iota$ and
 $F_1,\ldots,F_k$ of $\G$, can be computed in time
 $O(h(|\varphi|){\cdot}n{\cdot}d^{h(|\varphi|)})$, for a computable
 function $h$.

 For computing the unary relations $C_{P,j,t}$ of $\G$, start by initializing all
 of them to the empty set.
 For each $v_{(\ov{b},\iota)}\in V$ do the following:
 Compute (by using the relations $F_1,\ldots,F_k$) the tuple $\ov{b}$
 and the mapping $\iota$.
 Let $a$ be the first component of $\ov{b}$.
 Consider the $\tilde{r}$-neighborhood $\N^{\A}_{\tilde{r}}(a)$ of $a$
 in $\A$, for $\tilde{r}\deff \hat{r}{+}r$.
 For each $P=(P_1,\ldots,P_\ell)\in\Part$, each
 $j\in\set{1,\ldots,\ell}$ such that $\iota_{P_j}=\iota$,
 and each $t\in T_P$, check whether
 $\N^{\A}_{\tilde{r}}(a)\models \theta_{P,t,j}(\ov{b})$. If so, add
 the element $v_{(\ov{b},\iota)}$ to the relation $C_{P,j,t}$ of $\G$.
 (This is correct, since the formula $\theta_{P,j,t}$ is $r$-local
 around its free variables, and the radius of the neighborhood is
 large enough.)

 This way, $\G$'s relations $C_{P,j,t}$ can be computed in time
 $O(h(|\varphi|){\cdot}n{\cdot}d^{h(|\varphi|)})$, for a computable
 function $h$.

 To compute the $E$-relation of $\G$, note that for all tuples
 $(v_{(\ov{b},\iota)},v_{(\ov{c},\hat{\iota})})\in E^{\G}$, we have
 $\dist^{\A}(a,c_j)\leq (2k{+}1)(2r{+}1)$, for all
 components $c_j$ of $\ov{c}$, where $a$ is the first component of $\ov{b}$.
 Thus, the $E$-relation of $\G$ can be computed as follows:
 Start by initializing this relation to the empty set.
 For each $v_{(\ov{b},\iota)}\in V$ do the following:
 Compute (by using the relations $F_1,\ldots,F_k$) the tuple $\ov{b}$.
 Let $a$ be the first component of $\ov{b}$.
 Consider the $r'$-neighborhood $\N^{\A}_{r'}(a)$ of $a$
 in $\A$, for $r'\deff r+(2k{+}1)(2r{+}1)$.
 Use a brute-force algorithm to compute all tuples $\ov{c}$
 of elements in $\N^{\A}_{r'-r}(a)$, such that $|\ov{c}|\leq k$ and
 $\N^{\A}_{r'}(a)\models\gamma_{P_j}(\ov{c})$ for $P_j=\set{1,\ldots,|\ov{c}|}$.
 Check if there are components $b'$ of $\ov{b}$ and $c'$ of $\ov{c}$
 such that $\dist^{\N^{\A}_{r'}(a)}(b',c')\leq 2r{+}1$.
 If so,
 add to $E^{\G}$ the tuple
 $(v_{(\ov{b},\iota)},v_{(\ov{c},\hat{\iota})})$ for each injective
 mapping $\hat{\iota}:\set{1,\ldots,|\ov{c}|}\to\set{1,\ldots,k}$.

 This way, the $E$-relation of $\G$ can be computed in time
 $O(h(|\varphi|){\cdot}n{\cdot}d^{h(|\varphi|)})$, for a computable
 function $h$.

 In summary, we obtain that $\G$ is computable from $\A$ and $\varphi$
 within the desired time  bound.


 To finish the proof of
 item~\ref{item:G:prop-quantifier-elim-detailed}, we need to give an
 upper bound on the degree of $\G$.
 As noted above,
 $(v_{(\ov{b},\iota)},v_{(\ov{c},\hat{\iota})})\in E^{\G}$ implies that
 $\dist^{\A}(a,c_j)\leq r'$ for $r'\deff (2k{+}1)(2r{+}1)$, for all
 components $c_j$ of $\ov{c}$, where $a$ is the first component of $\ov{b}$.
 Thus, for each fixed $v_{(\ov{b},\iota)}\in V$, the number of
 elements $v_{(\ov{c},\hat{\iota})}$ such that $(v_{(\ov{b},\iota)},v_{(\ov{c},\hat{\iota})})\in E^\G$
 is at most
 \[
  k! \cdot \sum_{s=1}^k|\N^{\A}_{r'}(a)|^s
  \quad \leq \quad
  k!\cdot |\N^{\A}_{r'}(a)|^{k+1}
  \quad \leq \quad
  k!\cdot d^{(r'+1)(k+1)}.
 \]
 Thus, since $E^\G$ is symmetric, its degree is $\leq 2k! d^{(r'+1)(k+1)}$.

 Similarly, for each tuple $(v_{(\ov{b},\iota)},a)\in F_{i}^{\G}$ (with
 $i\in\set{1,\ldots,k}$) we know that $a$ is the $\iota^{-1}(i)$-th component of
 $\ov{b}$ and each component of $\ov{b}$ belongs to the
 $\hat{r}$-neighborhood of $a$ in $\A$, for $\hat{r}=k(2r{+}1)$.
 Thus, for each fixed $a\in A$, the number of elements $v_{(\ov{b},\iota)}\in V$ such that
 $(v_{(\ov{b},\iota)},a)\in F_i^{\G}$ is at most
 \ $
   k!\cdot\sum_{s=1}^k|\N^{\A}_{\hat{r}}(a)|^s
   \allowbreak
   \ \leq \
   k!\cdot  d^{(\hat{r}+1)(k+1)}.
 $
 In summary, we thus obtain that $\G$ is of degree at most
 $d^{h(|\varphi|)}$ for a computable function $h$.

 \medskip

 \noindent
 \emph{Step~5: Computing $f$ within the time bounds of Item~\ref{item:f:prop-quantifier-elim-detailed}.}

Recall that for $\ov{a}\in A^k$ we have
 \[
   f(\ov{a}) \ \deff \
   \big(\, v_{(\ov{a}_{P_1},\iota_{P_1})},\ldots,v_{(\ov{a}_{P_\ell},\iota_{P_\ell})},
   \dummy, \ldots, \dummy \, \big),
 \]
 for the unique partition $P=(P_1,\ldots,P_\ell)\in\Part$ such that
 $\A\models\rho_P(\ov{a})$. The number of $\dummy$-components in $f(\ov{a})$ is $(k{-}\ell)$.

  We first show how to compute $f(\bar a)$ from $\bar a$ in constant
  time. This is where we use $\e$.

 To compute the partition $P$ for a given tuple $\ov{a}=(a_1,\ldots,a_k)$, we can
 proceed as follows:
 Construct an undirected graph $H$ with vertex set  $\set{1,\ldots,k}$,
   where there is an edge between $i\neq j$ iff
   $\dist^{\A}(a_i,a_j)\leq 2r{+}1$.
This can be done as follows. Let $R$ be the binary relation over $\dom(\A)$
containing all pairs $(a,b)$ such that  $\dist^{\A}(a_i,a_j)\leq 2r{+}1$. As
$\A$ as degree at most $d$, the size of $R$ is bounded by $n {\cdot} d^{2r+2}$ and $R$ can
be computed by a brute-force algorithm in time $O(n \cdot d^{2r+2})$. Hence by
the Storing Theorem (Theorem~\ref{thm-storing-complete}), we can
compute a data structure in time
$O(n^{1+\e} \cdot d^{2r+2})$ such that afterwards we can test in time
depending only on $\e$ whether a given pair is in $R$ or not.

   Once $H$ is computed, we can compute its connected components in time
   depending only on $k$.
   Let $\ell$ be the number of connected components of $H$.
   For each $j\in\set{1,\ldots,\ell}$ let $P_j$ be vertex set of the
   $j$-th connected component, such that $\min P_j < \min P_{j+1}$ for
   all $j\in\set{1,\ldots,\ell{-}1}$.
 After having constructed the partition $P=(P_1,\ldots,P_\ell)$,
 further $O(k^2)$ steps suffice to construct the tuples $\ov{a}_{P_1}$,
 \ldots, $\ov{a}_{P_\ell}$, the mappings
 $\iota_{P_1}$, \ldots, $\iota_{P_\ell}$, and the
 according tuple $f(\ov{a})$.
Let $\zeta_P$ be the function associating to each  pair $(\ov{a}_{P_j},\iota_{P_j})$ the element
$v_{(\bar a_{P_j},\iota_{P_j})}$. The domain of $\zeta_P$ is at most
$n{\cdot} d^{k(2r+1)+1}$ and $\zeta_P$ can be computed in $O(n\cdot
d^{k(2r+1)+1})$ by a brute-force algorithm. Hence, using
Theorem~\ref{thm-storing-complete} we can compute in time $O(n^{1+\e}\cdot
d^{k(2r+1)+1})$ a data structure such that afterwards we can obtain the result of
the function $\zeta_P$ in time depending only on $\e$.

Altogether, after the preprocessing, we can compute $f(\ov{a})$ in time
$O(k^2)$.

Recall that using the relation $F$, it is straightforward to compute $f^{-1}(\ov{v})$
upon input of $\ov{v} \in \dom(\G)$ in time  $O(k^2)$.

This concludes the proof of proof of Proposition~\ref{prop-quantifier-elim-detailed}.
\end{proof}

\section{Conclusion}\label{section-conclusion}

For classes of databases of low degree, we presented an algorithm which
enumerates the answers to first-order queries with constant delay after
pseudo-linear preprocessing. An inspection of the proof shows that the
constants involved are non-elementary in the query size {(this is already the
case for Theorem~\ref{thm:grohe-low-degree} and we build upon this result, this
is also a consequence of Gaifman Normal Form which derives a new
formula of non-elementary size~\cite{DBLP:conf/icalp/DawarGKS07})}.

In
the bounded degree
case the constants are triply exponential in the query size~\cite{KS11}. In the
(unranked) tree case the constants are provably non-elementary~\cite{FrickG04}
(modulo some complexity assumption). We do not know what is the situation for
classes of low degree.

If the database is updated, for instance if a tuple is deleted or inserted, it would be
desirable to be able to update efficiently the data structure that is computed
for deriving in constant time  counting, testing, and enumeration. With the data
structure given in this paper it is not clear how to do this without
recomputing everything from scratch. However it has been shown recently that
there is another data structure, providing the same constant time properties that
furthermore can be updated in time $O(n^\e)$ upon insertion or
deletion of a tuple~\cite{DBLP:journals/corr/abs-2010-02982}.

It would also be interesting to know whether we can enumerate the answers to a
query using the \emph{lexicographical} order (as it is the case over structures
of bounded expansion~\cite{KazanaS13}).

\bibliographystyle{alphaurl}
\bibliography{biblioshort}

\end{document}

%% file: macro.tex
\newenvironment{mi}{\begin{itemize}}{\end{itemize}}
\newenvironment{me}{\begin{enumerate}[(1)]}{\end{enumerate}}

\renewcommand{\nc}{\newcommand}
\nc{\rnc}{\renewcommand}

\nc{\myparagraph}[1]{\paragraph{\bfseries #1.}}

\nc{\twodots}{.\,.\,}
\nc{\bs}[1]{\ensuremath{\boldsymbol{#1}}}
\rnc{\epsilon}{\varepsilon}
\rnc{\rho}{\varrho}

\nc{\ov}[1]{\ensuremath{\bar{#1}}}
\nc{\vecb}[1]{\ensuremath{\overleftarrow{#1}}} 
\nc{\vecf}[1]{\ensuremath{\overrightarrow{#1}}} 

\nc{\deff}{:=}

\nc{\und}{\ensuremath{\wedge}}
\nc{\Und}{\ensuremath{\bigwedge}}
\nc{\oder}{\ensuremath{\vee}}
\nc{\Oder}{\ensuremath{\bigvee}}
\nc{\nicht}{\ensuremath{\neg}}
\nc{\impl}{\ensuremath{\rightarrow}}
\nc{\gdw}{\ensuremath{\leftrightarrow}}

\nc{\true}{\ensuremath{\textit{true}}}
\nc{\false}{\ensuremath{\textit{false}}}

\nc{\ar}{\ensuremath{\textit{ar}}}
\nc{\qr}{\ensuremath{\textit{qr}}}
\nc{\free}{\ensuremath{\textit{free}}}

\nc{\size}[1]{\ensuremath{|\!|#1|\!|}}

\rnc{\geq}{\ensuremath{\geqslant}}
\rnc{\leq}{\ensuremath{\leqslant}}

\nc{\set}[1]{\ensuremath{\{ #1 \}}}
\nc{\setc}[2]{\set{#1 \,:\, #2}}
\nc{\bigset}[1]{\ensuremath{ \left\{ #1 \right\} } }
\nc{\bigsetc}[2]{\bigset{#1 \,:\, #2}}

\nc{\dom}{\ensuremath{\textit{dom}}}

\nc{\img}{\ensuremath{\textit{img}}}

\nc{\abgerundet}[1]{\ensuremath{
    \left\lfloor {#1} \right\rfloor }}
\nc{\aufgerundet}[1]{\ensuremath{\left\lceil {#1} \right\rceil }}

\nc{\NN}{\ensuremath{\mathbb{N}}}   
\nc{\NNpos}{\ensuremath{\NN_{\mbox{\tiny $\scriptscriptstyle
        \geq 1$}}}}

\nc{\ZZ}{\ensuremath{\mathbb{Z}}}
\nc{\RR}{\ensuremath{\mathbb{R}}}
\nc{\RRpos}{\ensuremath{\RR_{\mbox{\tiny $\scriptscriptstyle > 0$}}}}

\nc{\QQ}{\ensuremath{\mathbb{Q}}}
\nc{\QQpos}{\ensuremath{\QQ_{\mbox{\tiny $\scriptscriptstyle > 0$}}}}

\nc{\INT}{\textup{int}}
\nc{\intoo}[1]{\ensuremath{\INT(#1)}}
\nc{\intoc}[1]{\ensuremath{\INT(#1]}} 
\nc{\intco}[1]{\ensuremath{\INT[#1)}} 
\nc{\intcc}[1]{\ensuremath{\INT[#1]}}

\nc{\into}[1]{\intoo{#1}}
\nc{\intc}[1]{\intcc{#1}}

\nc{\Graphs}{\ensuremath{\mathscr{G}}}

\nc{\dist}{\ensuremath{\textit{dist}}}
\nc{\degree}{\ensuremath{\textit{degree}}}
\nc{\dd}{\ensuremath{d }}

\nc{\dummy}{\ensuremath{v_{\bot}}}

\nc{\A}{\ensuremath{\mathcal{A}}\xspace}

\nc{\Aq}{\ensuremath{\A\!\!\downarrow\!\!q}\xspace}

\nc{\B}{\ensuremath{\mathcal{B}}\xspace}
\nc{\G}{\ensuremath{\mathcal{G}}\xspace}

\nc{\N}{\ensuremath{\mathcal{N}}} 

\nc{\Class}{\ensuremath{C}\xspace} 

\nc{\DataStructureFont}[1]{\ensuremath{\texttt{#1}}}
\nc{\AL}{\DataStructureFont{AL}} 
\nc{\DM}{\DataStructureFont{DM}} 
\nc{\NA}{\DataStructureFont{NA}} 
\nc{\DA}{\DataStructureFont{DA}} 
\nc{\NCList}{\DataStructureFont{NC}} 
\nc{\CList}{\DataStructureFont{C}}   
\nc{\Skip}{\DataStructureFont{Skip}} 
\nc{\First}{\DataStructureFont{First}}
\nc{\Next}{\DataStructureFont{Next}}
\nc{\Last}{\DataStructureFont{Last}}

\nc{\nil}{\DataStructureFont{NIL}}

\nc{\R}{\ensuremath{{\hat{r}}}} 

\nc{\class}[1]{\ensuremath{\textrm{\upshape{#1}}}}
\nc{\FO}{\class{FO}\xspace}
\nc{\ACzero}{\ensuremath{\class{AC}^0}}
\nc{\ACO}{\ACzero}

\nc{\CDal}{\ensuremath{\textsc{Constant-Delay}_{\textup{almost-lin}}}}

\nc{\Algo}{\ensuremath{\mathfrak{A}}}
\nc{\Problem}[1]{\ensuremath{\textsc{#1}}}

\nc{\Part}{\ensuremath{\mathcal{P}}} 

\rnc{\next}{\ensuremath{\textit{next}}}
\nc{\first}{\ensuremath{\textit{first}}}
\nc{\void}{\ensuremath{\textit{void}}\xspace}

